\documentclass[10pt,final,twocolumn]{IEEEtran}

\makeatletter
\newif\if@restonecol
\makeatother

\usepackage{eucal}
\usepackage{graphicx}
\usepackage{times}
\usepackage{latexsym}
\usepackage{amsmath}
\usepackage{amssymb}
\usepackage{amsthm}
\usepackage{epsfig}
\usepackage{cite}
\usepackage{cases}
\usepackage{amsfonts}
\usepackage{indentfirst}
\usepackage{epsfig}
\usepackage{amsopn}
\usepackage{bm}
\usepackage[linesnumbered,ruled]{algorithm2e}
\usepackage{subfig}
\usepackage{eufrak}
\usepackage{stfloats}
\usepackage{booktabs}
\usepackage{slashbox}

\newtheorem{lemma}{Lemma}
\newtheorem{theorem}{Theorem}

\newtheorem{comment}{Comment}[section]
\newcommand{\Tr}{\mathsf{Tr}}
\newcommand{\Real}{\mathsf{Re}}

\begin{document}

\title{\huge Joint Transceiver Design for Wireless Sensor Networks through Block Coordinate Descent Optimization}

\author{Yang Liu$^{\dagger}$,\ \ Jing Li$^{\dagger}$,\ \ Xuanxuan Lu$^{\dagger}$, \ \  and \ \ Chau Yuen$^{\ddagger}$ \\
$^{\dagger}$ Electrical and Computer Engineering Department,
 Lehigh University, Bethlehem, PA 18015, USA\\
$^{\ddagger}$ Singapore University of Technology and Design, 20 Dover Drive, 138682, Singapore \\
Email: \{yal210@lehigh.edu, jingli@ece.lehigh.edu, xul311@lehigh.edu, yuenchau@sutd.edu.sg\}
\vspace*{-0.5cm}
\thanks{Supported by National Science Foundation under Grants No.0928092, 1133027 and 1343372.}}


\maketitle


\begin{abstract}

This paper considers the joint transceiver design in a wireless sensor network where multiple sensors observe the same physical event and transmit their contaminated observations to a fusion center, with all nodes equipped with multiple antennae and linear filters. Under the mean square error (MSE) criterion, the joint beamforming design problem can be formulated as a nonconvex optimization problem. To attack this problem, various block coordinate descent (BCD) algorithms are proposed with convergence being carefully examined. First we propose a two block coordinate descent (2-BCD) algorithm that iteratively designs all the beamformers and the linear receiver, where both subproblems are convex and the convergence of limit points to stationary points is guaranteed. Besides, the thorough solution to optimizing one single beamformer is given, which, although discussed several times, is usually incomplete in existing literature. Based on that, multiple block coordinate descent algorithms are proposed. Solving the joint beamformers' design by cyclically updating each separate beamformer under the 2-BCD framework gives birth to a layered BCD algorithm, which guarantees convergence to stationary points. Besides that, a wide class of multiple BCD algorithms using the general essentially cyclic updating rule has been studied. As will be seen, by appropriately adjusting the update of single beamformer, fast converging, highly efficient and stationary point achieving algorithms can be obtained. Extensive numerical results are presented to verify our findings.

\end{abstract}



\section{Introduction}
\label{sec:introduction}

Consider a typical wireless sensor network (WSN) comprised of a fusion center (FC)  and numerous sensors that are spatially distributed and wirelessly connected to provide surveillance to the same physical event. After harvesting information from the environment, these sensors transmit distorted observations to the fusion center (FC) to perform data fusion. A central underlying problem is how to design the sensors and the fusion center to collaboratively accomplish sensing, communication and fusion task in an efficient and trust-worthy manner. 

When the sensors and the fusion center are all equipped with multiple antennas and linear filters, this problem may be regarded as one of the cooperative multi-input multi-output (MIMO) beamforming design problems, which have been tackled from various perspectives \cite{bib:Sensor_compress_1, bib:Sensor_compress_2, bib:FangLi_2, bib:Sensor_compress_3, bib:sensor_network_Cui, bib:sensor_network_Hamid, bib:JunFang_MI, bib:J2_Yang_to_be_submitted, bib:J3_Yang_to_be_submitted}. For example \cite{bib:Sensor_compress_1, bib:Sensor_compress_2,  bib:Sensor_compress_3, bib:FangLi_2} target compression (dimensionality reduction) beamforming. \cite{bib:Sensor_compress_1} and \cite{bib:Sensor_compress_2} consider the scenarios where the orthogonal multiple access channels (MAC) between the sensors and the fusion center are perfect without fading or noise. For wireless communication, the assumption of ideal channel is unrealistic and the imperfect channels are considered in \cite{bib:FangLi_2, bib:Sensor_compress_3, bib:sensor_network_Cui, bib:sensor_network_Hamid, bib:JunFang_MI, bib:J2_Yang_to_be_submitted, bib:J3_Yang_to_be_submitted}. \cite{bib:FangLi_2} researches the problem of scalar source transmission with all sensors sharing one total transmission power and using orthogonal MAC. Imperfect coherent MAC and separate power constraint for each sensor are considered in \cite{bib:Sensor_compress_3}, under the assumptions that all channel matrices are square and nonsingular. The work \cite{bib:sensor_network_Cui} and \cite{bib:sensor_network_Hamid} are particularly relevant to our problem. \cite{bib:sensor_network_Cui} is the first to present a very general system model, which considers noisy and fading channels, separate power constraints and does not impose any constraints on the dimensions of beamformers or channel matrices. \cite{bib:sensor_network_Cui} provides the solutions to several interesting special cases of the general model for coherent MAC, such as the noiseless channel case and the no-intersymbol-interference (no-ISI) channel case. In \cite{bib:sensor_network_Hamid}, the authors develop a useful type of iterative method that is applicable to the general model in \cite{bib:sensor_network_Cui} for coherent MAC. 
All the works mentioned above take the mean square error (MSE) as performance metric. Recently, under the similar system settings of \cite{bib:sensor_network_Cui}, joint transceiver design to maximize mutual information(MI) attract attentions and are studied in \cite{bib:JunFang_MI} and \cite{bib:J2_Yang_to_be_submitted}, with orthogonal and coherent MAC being considered respectively. The SNR maximization problem for wireless sensor network with coherent MAC is reported in \cite{bib:J3_Yang_to_be_submitted}.

It is interesting to note that the beamforming design problems in MIMO multi-sensor decision-fusion system have significant relevance with those in other multi-agent communication networks, e.g. MIMO multi-relay and multiuser communication systems. A large number of exciting papers exist in the literature, see, for example, \cite{bib:MIMO_relay_unifying, bib:MIMO_multiuser, bib:MIMO_relay_multiuser, bib:WMMSE_QingjiangShi} and the references therein.


This paper considers the very general coherent MAC model discussed in \cite{bib:sensor_network_Cui,bib:sensor_network_Hamid}. To solve the original nonconvex joint beamforming problem, we propose several iterative optimization algorithms using the block coordinate descent (BCD) methodology, with their convergence and complexity carefully studied. Specifically our contributions include:

1) We first propose a 2 block coordinate descent (2-BCD) method that decomposes the original problem into two subproblems--- one subproblem, with all the beamformers given, is a linear minimum mean square error (LMMSE) filtering problem  and the other one, jointly optimizing the beamformers with the receiver given, is shown to be convex. It is worth mentioning that \cite{bib:sensor_network_Cui} considers the special case where the sensor-FC channels are intersymbol-interference (ISI) free (i.e. the sensor-FC channel matrix is an identity matrix) and solves the entire problem by semidefinite programming(SDP) and relaxation. Here we reformulate the joint optimization of beamformers, even with arbitrary sensor-FC channel matrices, into a second-order cone programming(SOCP) problem, which is more efficiently solvable than the general SDP problem. Convergence analysis shows that this 2-BCD algorithm guarantees its limit points to be stationary points of the original problem. Interestingly enough, although not presented in this article, the proposed 2-BCD algorithm has one more fold of importance---the convexity of its subproblem jointly optimizing beamformers can be taken advantage of by the multiplier method \cite{bib:Distributed_Palomar}, which requires the original problem to be convex, and therefore gives birth to decentralized solutions to the problem under the 2-BCD framework. 

2) We have also attacked the MSE minimization with respect to one single beamformer and developed fully analytical solutions (possibly up to a simple one-dimension bisection search). 
It should be pointed out that, although the same problem has been studied in several previous papers (e.g. \cite{bib:MIMO_multiuser, bib:MIMO_relay_multiuser, bib:sensor_network_Hamid, bib:WMMSE_QingjiangShi}), we are able to carry out the analysis to the very end and thoroughly solved the problem by clearly describing the solution structure and deriving the solutions for all possible cases. Specifically,  we explicitly obtain the conditions for judging the positiveness of the Lagrange multiplier. Moreover, in the zero-Lagrange-multiplier case with singular quadratic matrix, we give out the energy-preserving solution via pseudoinverse among all possible optimal solutions. To the best of our knowledge, these exact results have never been discussed in existing literature.

3) Our closed form solution for one single beamformer's update paves the way to multiple block coordinate descent algorithms. A layered-BCD algorithm is proposed, where an inner-loop cyclically optimizing each separate beamformer is embedded in the 2-BCD framework. This layered-BCD algorithm is shown to guarantee the limit points of its solution sequence to be stationary. Besides we also consider a wide class of multiple block coordinate descent algorithms with the very general essentially cyclic updating rule. It is interesting to note that this class of algorithms subsumes the one proposed in \cite{bib:sensor_network_Hamid} as a specialized realization. Furthermore, as will be shown, by appropriately adjusting the update of each single beamformer to a proximal version and introducing approximation, the essentially cyclic multiple block coordinate descent algorithm exhibits fast converging rate, guarantees convergence to stationary points and achieves high computation efficiency.

The rest of the paper is organized as follows: Section \ref{sec:system model} introduces the system model of the joint beamforming problem in the MIMO wireless sensor network. Section \ref{sec:two_block} discusses the 2-BCD beamforming design approach and analyzes its convexity and convergence. Section \ref{sec:further_decoupling} discusses the further decomposition of the joint optimization of beamformers, including the closed form solution to one separate beamformer's update, layered BCD algorithms, essentially cyclic BCD algorithms and their variants and convergence. Section \ref{sec:numerical results} provides simulation verification and Section \ref{sec:conclusion} concludes this article.

\emph{Notations}: We use bold lowercase letters to denote complex vectors and bold capital letters to denote complex matrices. $\mathbf{0}$, $\mathbf{O}_{m\times n}$, and $\mathbf{I}_m$ are used to denote zero vectors, zero matrices of dimension $m\times n$, and identity matrices of order $m$ respectively. $\mathbf{A}^T$, $\mathbf{A}^{\ast}$ and $\mathbf{A}^H$ are used to denote transpose, conjugate and conjugate transpose (Hermitian transpose) respectively of an arbitrary complex matrix $\mathbf{A}$. $\Tr\{\cdot\}$ denotes the trace operation of a square matrix. $|\cdot|$ denotes the modulus of a complex scalar, and $\|\cdot\|_2$ denotes the $l_2$-norm of a complex vector. $vec(\cdot)$ means vectorization operation of a matrix, which is performed by packing the columns of a matrix into a long one column. $\otimes$ denotes the Kronecker product. $\mathsf{diag}\{\mathbf{A}_1,\cdots,\mathbf{A}_n\}$ denotes the block diagonal matrix with its $i$-th diagonal block being the square complex matrix $\mathbf{A}_i$, $i\in\{1,\cdots,n\}$. $\mathsf{Re}\{x\}$ denotes the real part of a complex value $x$.

\section{System Model}
\label{sec:system model}

Consider a centralized wireless sensor network with $L$ sensors and one fusion center where all the nodes are equipped with multiple antennae, as shown in Figure \ref{fig:sysmodel}. Let $M$ and $N_i$ ($i=1,2,\cdots,L$) be the number of antennas provisioned to the fusion center and the $i$-th sensor respectively. Denote $\mathbf{s}$ as the common source vector observed by all sensors. The source $\mathbf{s}$ is a complex vector of dimension $K$, i.e. $\mathbf{s}\in\mathbb{C}^{K\times1}$, and is observed by all the sensors. At the $i$-th sensor, the source signal is linearly transformed by an observation matrix $\mathbf{K}_i\in\mathbb{C}^{J_i\times K}$ and corrupted by additive observation noise $\mathbf{n}_i$, which has zero mean and covariance matrix $\mathbf{\Sigma}_i$.   

\begin{figure}[htb]
\centerline{
\includegraphics[width=0.48\textwidth]{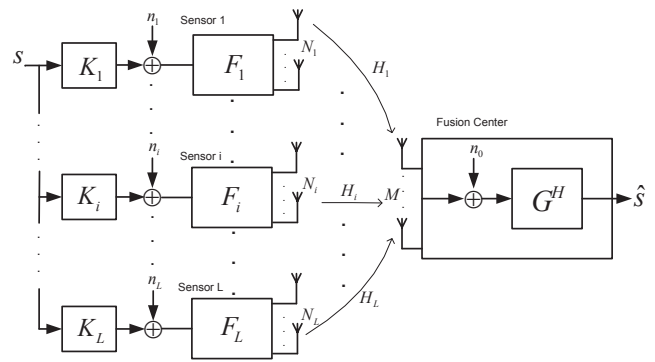}
}
\caption{Multi-Sensor System Model}
\label{fig:sysmodel}
\end{figure}

Each sensor applies some linear precoder, $\mathbf{F}_i\in\mathbb{C}^{N_i\times J_i}$, to its observation $(\mathbf{K}_i\mathbf{s}+\mathbf{n}_i)$ before sending it to the common fusion center. Denote $\mathbf{H}_i\in\mathbb{C}^{M\times N_i}$ as the fading channel between the $i$-th sensor and the fusion center. Here we considers the coherent MAC model, where the transmitted data is superimposed and corrupted by additive noise at the fusion center. Without loss of generality, the channel noise is modeled as a vector $\mathbf{n}_0\in\mathbb{C}^{M\times1}$ with zero mean and white covariance $\sigma_0^2\mathbf{I}_{M}$. The fusion center, after collecting all the results, applies a linear postcoder, $\mathbf{G}^H\in\mathbb{C}^{K\times M}$, to retrieve the original source $\mathbf{s}$.


This system model depicted in Figure \ref{fig:sysmodel} is the same as the general model presented in \cite{bib:sensor_network_Cui,bib:sensor_network_Hamid}. Following their convention, we assume that the system is perfectly time-synchronous (which may be realized via the GPS system) and that all the channel state information $\mathbf{H}_i$ is known (which may be achieved via channel estimation techniques). Since the sensors and the fusion center are usually distributed over a wide range of space, it is reasonable to assume that the noise $\mathbf{n}_i$ at different sensors and $\mathbf{n}_0$ at the fusion center are mutually uncorrelated.

The signal transmitted by the $i$-th sensor takes the form of $\mathbf{F}_i(\mathbf{K}_i\mathbf{s}+\mathbf{n}_i)$. The output $\hat{\mathbf{s}}$ of the postcoder at the fusion center is given as
\begin{align}
\!\!\!\!\hat{\mathbf{s}}&=\mathbf{G}^H\mathbf{r}=\mathbf{G}^H\bigg(\sum_{i=1}^L\mathbf{H}_i\mathbf{F}_i(\mathbf{K}_i\mathbf{s}+\mathbf{n}_i)+\mathbf{n}_0\bigg) \\
\!\!\!\!&=\mathbf{G}^H\bigg(\sum_{i=1}^L\mathbf{H}_i\mathbf{F}_i\mathbf{K}_i\bigg)\mathbf{s}+\mathbf{G}^H\bigg(\underbrace{\sum_{i=1}^L\mathbf{H}_i\mathbf{F}_i\mathbf{n}_i+\mathbf{n}_0}_{\mathbf{n}} \bigg),\label{eq:sysmodel}
\end{align}
where the compound noise vector $\mathbf{n}$ has covariance matrix $\mathbf{\Sigma}_{\mathbf{n}}$ given by 
\begin{align}
\mathbf{\Sigma}_{\mathbf{n}}=\sigma_0^2\mathbf{I}_M+ \sum_{i=1}^L\mathbf{H}_i\mathbf{F}_i\mathbf{\Sigma}_i\mathbf{F}_i^H\mathbf{H}_i^H. 
\label{eq:cov_n0}
\end{align}


In this paper, we take the mean square error as a figure of merit. The mean square error matrix $\mathbf{\Phi}$ is defined as
\begin{align}
\mathbf{\Phi}&\triangleq\mathsf{E}\big\{\big(\mathbf{s}-\hat{\mathbf{s}}\big)\big(\mathbf{s}-\hat{\mathbf{s}}\big)^H\big\}. \label{eq:MSE_phi_definition}
\end{align}
Assume that the source signal $\mathbf{s}$ has zero mean and a covariance matrix $\mathbf{\Sigma}_{\mathbf{s}}\triangleq\mathsf{E}\{\mathbf{s}\mathbf{s}^H\}$. By plugging (\ref{eq:sysmodel}) into (\ref{eq:MSE_phi_definition}), we can express the $\mathsf{MSE}$ matrix $\mathbf{\Phi}$ as a function of $\{\mathbf{F}_i\}$ and $\mathbf{G}$ as:
\begin{align}
\!\!\!\!\!\!\!\mathbf{\Phi}\Big(\{\mathbf{F}_i\}_{i=1}^{L},\mathbf{G}\Big)\!\!&=\!\!\mathbf{G}^H\Big(\!\sum_{i=1}^L\mathbf{H}_i\mathbf{F}_i\mathbf{K}_i\!\Big)\mathbf{\Sigma}_{\mathbf{s}}\Big(\!\sum_{i=1}^L\mathbf{H}_i\mathbf{F}_i\mathbf{K}_i\!\Big)^H\mathbf{G} \nonumber\\
&\!\!\!\!\!\!-\!\!\mathbf{G}^H\!\Big(\!\sum_{i=1}^L\mathbf{H}_i\mathbf{F}_i\mathbf{K}_i\!\Big)\mathbf{\Sigma}_{\mathbf{s}}\!\!-\!\!\mathbf{\Sigma}_{\mathbf{s}}\Big(\!\sum_{i=1}^L\mathbf{H}_i\mathbf{F}_i\mathbf{K}_i\!\Big)^H\!\!\mathbf{G}\!\!\nonumber\\
&\!\!\!\!\!\!\!\!+\!\!\sum_{i=1}^L\!\mathbf{G}^H\mathbf{H}_i\mathbf{F}_i\mathbf{\Sigma}_{i}\mathbf{F}_i^H\mathbf{H}_i^H\mathbf{G}\!+\!\sigma_0^2\mathbf{G}^H\mathbf{G}\!+\!\!\mathbf{\Sigma}_{\mathbf{s}}.\label{eq:MSE_matrix}
\end{align}
The total $\mathsf{MSE}$ is then given by
\begin{align}
\mathsf{MSE}\Big(\big\{\mathbf{F}_i\big\}_{i=1}^{L},\mathbf{G}\Big)\triangleq \mathsf{Tr}\Big\{\mathbf{\Phi}\Big(\big\{\mathbf{F}_i\big\}_{i=1}^{L},\mathbf{G}\Big)\Big\}\label{eq:MSE_def}.
\end{align}

We consider the case where each sensor has its own transmission power constraint.
This means $\mathsf{E}\{\|\mathbf{F}_i(\mathbf{K}_i\mathbf{s}+\mathbf{n}_i)\|_2^2\}=\mathsf{Tr}\{\mathbf{F}_i(\mathbf{K}_i\mathbf{\Sigma_{\mathbf{s}}}\mathbf{K}_i^H\!\!+\!\!\mathbf{\Sigma}_{i})\mathbf{F}_i^H\}\leq P_i$. 
The overall beamforming design problem can then be formulated as the following optimization problem:
\begin{subequations}
\begin{align}
&\!\!\!\!\!(\mathsf{P}0)\!:\!\!\!\!\underset{\{\mathbf{F}_i\}_{i=1}^{L},\mathbf{G}}{\min.}\!\!\mathsf{MSE}\big(\{\mathbf{F}_i\}_{i=1}^{L},\mathbf{G}\big), \label{eq:opt_prob0}\\
&\!\!s.t.\ \mathsf{Tr}\big\{\mathbf{F}_i(\mathbf{K}_i\mathbf{\Sigma_{\mathbf{s}}}\mathbf{K}_i^H\!\!+\!\!\mathbf{\Sigma}_{i})\mathbf{F}_i^H\big\}\!\leq\!P_i,\ i\in\{1,\cdots, L\}. \label{eq:opt_contr0}
\end{align}
\end{subequations}
The above problem is nonconvex, which can be verified by checking the special case where $\{\mathbf{F}_i\}_{i=1}^L$ and $\mathbf{G}$ are all scalars. 

The following of this paper consults to \emph{block coordinate descent} (BCD) method \cite{bib:parallel_computing_Bertsekas, bib:nonlinear_programming_Bertsekas, bib:Gauss-Seidel, bib:block_coordinate_decent_algorithm}, which is also known as Gauss-Seidel method, to solve ($\mathsf{P}0$) by partitioning the whole variables into separate groups and optimize each group (with the others being fixed) in an iterative manner. Appropriate decomposition can lead to efficiently solvable subproblems and may also provide opportunities for parallel computation.


\section{Two-Block Coordinate Descent (2-BCD)}
\label{sec:two_block}

In this section, we study a two block coordinate descent (2-BCD) method that decouples the design of the postcoder $\mathbf{G}$ (conditioned on the precoders), thereafter referred to as ($\mathsf{P}1$), from the design of all the precoders $\{\mathbf{F}_i\}_{i=1}^L$ (conditioned on the postcoder), thereafter referred to as ($\mathsf{P}2$).  

\subsection{($\mathit{P1}$): Optimizing $\mathbf{G}$ given $\{\mathbf{F}_i\}$}

For any given $\{\mathbf{F}_i\}_{i=1}^L$, minimizing $\mathsf{MSE}$ with respective to $\mathbf{G}$ becomes a strictly convex non-constrained quadratic problem ($\mathsf{P}1$):
\begin{align}
\!\!\!\!(\mathsf{P}1):\underset{\mathbf{G}}{\min}\ \mathsf{Tr}\Big\{\mathbf{\Phi}\Big(\mathbf{G} \Big|\big\{\mathbf{F}_i\big\}_{i=1}^L\Big)\Big\}.\label{eq:opt_prob1}
\end{align}
By equating the derivative $\frac{\partial}{\partial\mathbf{G}^*}\mathsf{MSE}\big(\mathbf{G}\big)$ with zero, the optimal receiver is readily obtained as the well-known Wiener filter \cite{bib:AdaptiveFilter} 
\begin{align}
\!\!\!\!\!\!\mathbf{G}_{(\mathsf{P}1)}^{\star}\!\!=\!\!\bigg[\Big(\!\sum_{i=1}^{L}\!\mathbf{H}_i\mathbf{F}_i\mathbf{K}_i\!\Big)\mathbf{\Sigma}_{\mathbf{s}}\!\Big(\!\sum_{i=1}^{L}\!\mathbf{H}_i\mathbf{F}_i\mathbf{K}_i\!\Big)^{\!\!H}\!\!\!\!+\!\!\mathbf{\Sigma}_{\mathbf{n}}\!\bigg]^{\!-\!1}\!\!\!\!\Big(\!\sum_{i=1}^{L}\!\mathbf{H}_i\mathbf{F}_i\mathbf{K}_i\!\Big)\mathbf{\Sigma}_{\mathbf{s}},\label{eq:G_MMSE}
\end{align}
where $\mathbf{\Sigma}_{\mathbf{n}}$ is given in (\ref{eq:cov_n0}).

\subsection{($\mathit{P2}$): Optimizing $\{\mathbf{F}_i\}$ given $\mathbf{G}$}

With $\mathbf{G}$ being fixed, the subproblem ($\mathsf{P}2$) minimizes $\mathsf{MSE}$ with respect to $\{\mathbf{F}_i\}_{i=1}^L$ is formulated as
\begin{subequations}
\begin{align}
&\!\!\!\!\!\!\!\!\!\!(\mathsf{P}2):\!\underset{\{\mathbf{F}_i\}_{i=1}^{L}}{\min.} \mathsf{Tr}\Big\{\mathbf{\Phi}\Big(\big\{\mathbf{F}_i\big\}_{i=1}^L\Big|\mathbf{G}\Big)\Big\}, \label{eq:opt_prob2} \\
&\!\!s.t.\ \mathsf{Tr}\big\{\mathbf{F}_i(\mathbf{K}_i\mathbf{\Sigma}_{\mathbf{s}}\mathbf{K}_i^H\!\!+\!\!\mathbf{\Sigma}_i)\mathbf{F}_i^H\big\}\!\leq\! P_i, \  i\!\in\!\{1,\cdots, L\}. \label{eq:opt_contr2}
\end{align}
\label{eq:opt_prob_P2}
\end{subequations}
Below we discuss the convexity of ($\mathsf{P}2$). 

\begin{theorem}\label{thm:P_2_convexity}
($\mathit{P2}$) is convex with respect to $\{\mathbf{F}_i\}_{i=1}^{L}$.
\end{theorem}
\begin{proof}
First consider the function $f\big(\mathbf{X}\big):\mathbb{C}^{m\times n}\mapsto\mathbb{R}, f(\mathbf{X})\!=\!\mathsf{Tr}\{\mathbf{A}^H\mathbf{X}\mathbf{\Sigma}\mathbf{X}^H\mathbf{A}\}$, 
where the constant matrices $\mathbf{A}$ and $\mathbf{\Sigma}$ have appropriate dimensions and $\mathbf{\Sigma}$ is Hermitian and positive semidefinite.

By the identities $\mathsf{Tr}\{\mathbf{AB}\}=\mathsf{Tr}\{\mathbf{BA}\}$ and $\mathsf{Tr}\{\mathbf{ABCD}\}=vec^T(\mathbf{D}^T)\big[\mathbf{C}^T\otimes\mathbf{A}\big]vec(\mathbf{B})$, $f\big(\mathbf{X}\big)$ can be equivalently written as $f(\mathbf{X})=vec^H(\mathbf{X})[\mathbf{\Sigma}^*\otimes(\mathbf{A}\mathbf{A}^H)]vec(\mathbf{X})$. 

According to \cite{bib:KroneckerEigenvalue}, i) $[\mathbf{A}\!\otimes\!\mathbf{B}]^H\!=\!\mathbf{A}^H\!\otimes\!\mathbf{B}^H$; ii) for any two Hermitian matrices $\mathbf{A}_{m\times m}$ and $\mathbf{B}_{n\times n}$ having eigenvalues $\{\lambda_{i}(\mathbf{A})\}_{i=1}^{m}$ and $\{\lambda_{j}(\mathbf{B})\}_{j=1}^{n}$ respectively, the eigenvalues of their Kronecker product $\mathbf{A}\!\otimes\!\mathbf{B}$ are given by $\{\lambda_{i}(\mathbf{A})\lambda_{j}(\mathbf{B})\}_{i=1,j=1}^{m,n}$. As a result, $\mathbf{A}\otimes\mathbf{B}$ is positive semidefinite when $\mathbf{A}$ and $\mathbf{B}$ are positive semidefinite. 
 
Since $\mathbf{A}\mathbf{A}^H$ and $\mathbf{\Sigma}^*$ are both positive semidefinite, $[\mathbf{\Sigma}^*\otimes(\mathbf{A}\mathbf{A}^H)]$ is positive semidefinite and therefore $f(\mathbf{X})$ is actually a convex homogeneous quadratic function of $vec\big(\mathbf{X}\big)$.

Now substitute $\mathbf{X}$ in $f(\mathbf{X})$ by $\sum_{i=1}^L\big(\mathbf{H}_i\mathbf{F}_i\mathbf{K}_i\big)$ 
and recall the fact that affine operation preserves convexity \cite{bib:CvxOpt}, the term  $\Tr\Big\{\mathbf{G}^H\big(\!\sum_{i=1}^L\mathbf{H}_i\mathbf{F}_i\mathbf{K}_i\big)\mathbf{\Sigma}_{\mathbf{s}}\big(\!\sum_{i=1}^L\mathbf{H}_i\mathbf{F}_i\mathbf{K}_i\big)^H\mathbf{G}\Big\}$ in the objective function ($\mathsf{P}2$) is therefore convex with respect to $\{\mathbf{F}_i\}_{i=1}^{L}$. By the same reasoning, the remaining terms in the objective and the constraints of ($\mathsf{P}2$) are either convex quadratic or affine functions of $\{\mathbf{F}_i\}_{i=1}^{L}$ and therefore the problem $(\mathsf{P}2)$ is convex with respective to $\{\mathbf{F}_i\}_{i=1}^L$.
\end{proof}

In the following we reformulate the subproblem ($\mathsf{P}2$) into a standard second order cone programming(SOCP) presentation. To this end, we introduce the following notations:
\begin{subequations}
\begin{align}
\mathbf{f}_i&\triangleq vec\big(\mathbf{F}_i\big);\ \ \ 
\mathbf{g}\triangleq vec\big(\mathbf{G}\big);\\
\mathbf{A}_{ij}&\triangleq(\mathbf{K}_j\mathbf{\Sigma}_{\mathbf{s}}\mathbf{K}_i^H)^T\otimes\Big(\mathbf{H}^H_{i}\mathbf{G}\mathbf{G}^H\mathbf{H}_j\Big);\\
\mathbf{B}_i&\triangleq(\mathbf{K}_i\mathbf{\Sigma}_{\mathbf{s}})^T\otimes\mathbf{H}_i;\\  
\mathbf{C}_i&\triangleq\mathbf{\Sigma}^*_i\otimes\Big(\mathbf{H}_i^H\mathbf{G}\mathbf{G}^H\mathbf{H}_i\Big).   \label{eq:notation_ABC}
\end{align}
\end{subequations}
By the identity $\mathsf{Tr}\{\mathbf{ABCD}\}=vec^T(\mathbf{D}^T)\big[\mathbf{C}^T\otimes\mathbf{A}\big]vec(\mathbf{B})$ and the above notations, we can rewrite the $\mathsf{MSE}$ in ($\mathsf{P}2$) as 
\begin{align}
\mathsf{MSE}\Big(\big\{\mathbf{f}_i\big\}_{i=1}^{L}\Big|\mathbf{g}\Big)&=\sum_{i=1}^{L}\sum_{j=1}^{L}\mathbf{f}_i^H\mathbf{A}_{ij}\mathbf{f}_j-2\mathsf{Re}\Big(\sum_{i=1}^L\mathbf{g}^H\mathbf{B}_i\mathbf{f}_i\Big) \nonumber \\
&\!\!\!\!+\sum_{i=1}^{L}\mathbf{f}_i^H\mathbf{C}_i\mathbf{f}_i+\sigma_0^2\|\mathbf{g}\|^2+\Tr\big\{\mathbf{\Sigma}_{\mathbf{s}}\big\}.\label{eq:MSE_def_equiv_2}
\end{align}

By further denoting
\begin{subequations}
\begin{align}
\!\!\!\!\!\!\!\!\!\!\!\!\!\!\!\!\!\!\mathbf{f}^T&\triangleq\big[\mathbf{f}_1^T, \cdots, \mathbf{f}_i^T, \cdots, \mathbf{f}_L^T\big]; \\
\mathbf{A}&\triangleq
\left[
\begin{array}{llll}
\mathbf{A}_{1,1} & \mathbf{A}_{1,2} & \cdots & \mathbf{A}_{1,L} \\
\mathbf{A}_{2,1} & \mathbf{A}_{2,2} & \cdots & \mathbf{A}_{2,L} \\
\vdots & \vdots & \ddots & \vdots \\
\mathbf{A}_{L,1} & \mathbf{A}_{L,2} & \cdots & \mathbf{A}_{L,L}
\end{array}
\right]; \\
\mathbf{B}&\triangleq\Big[\mathbf{B}_1,\cdots,\mathbf{B}_i,\cdots, \mathbf{B}_L\Big]; \\
\mathbf{C}&\triangleq \mathsf{diag}\Big\{\mathbf{C}_1, \cdots, \mathbf{C}_i, \cdots, \mathbf{C}_L\Big\}; \\
\mathbf{D}_i&\triangleq \mathsf{diag}\Big\{\mathbf{O}_{\sum_{j=1}^{i-1}J_jN_j}, \mathbf{E}_i, \mathbf{O}_{\sum_{j=i+1}^{L}J_jN_j}\Big\},\nonumber\\
&\qquad\qquad\qquad\qquad\qquad\qquad\quad i\in\{1,\cdots,L\};\\
\mathbf{E}_i&\triangleq\big(\mathbf{K}_i\mathbf{\Sigma}_{\mathbf{s}}\mathbf{K}_i^H\!\!+\!\!\mathbf{\Sigma}_i\big)^T\!\otimes\!\mathbf{I}_{N_i},\ \ \ \ \ i\in\{1,\cdots,L\};\\
c&\triangleq \Tr\{\mathbf{\Sigma}_{\mathbf{s}}\}+\sigma_0^2\|\mathbf{g}\|^2,\label{eq:P2_parameters}
\end{align}
\end{subequations}
the problem ($\mathsf{P}2$) can be rewritten as ($\mathsf{P}2^{\prime}$):
\begin{subequations}
\begin{align}
\!\!\!\!\!\!(\mathsf{P}2^{\prime}):\underset{\mathbf{f}}{\min}&\ \mathbf{f}^H\big(\mathbf{A}\!\!+\!\!\mathbf{C}\big)\mathbf{f}\!-\!2\Real\{\mathbf{g}^H\mathbf{B}\mathbf{f}\}\!+\!c, \label{eq:opt_prob2}\\
s.t.& \ \ \mathbf{f}^H\mathbf{D}_i\mathbf{f}\leq P_i, \ \ \ \ \ \ i\in\{1,\cdots, L\}. \label{eq:opt_contr2}
\end{align}\label{eq:opt_prob_P2_prime}
\end{subequations}
As proved by Theorem \ref{thm:P_2_convexity}, ($\mathsf{P}2^{\prime}$) (or equivalently   ($\mathsf{P}2$)) is convex, which implies $(\mathbf{A}\!+\!\mathbf{C})$ is positive semidefinite. Thus the square root $(\mathbf{A}\!+\!\mathbf{C})^{\frac{1}{2}}$ exists. The above problem can therefore be reformulated in an SOCP form as follows
\begin{subequations}
\begin{align}
&\!\!\!\!\!\!\!\!\!\!\!\!\!\!\!\!(\mathsf{P}2_{SOCP}): \underset{\mathbf{f},t,s}{\min.} \ t, \\
&\!\!\!\!\!\!\!\!\!\!\!\!\mathrm{s.t.}\ s-2\mathsf{Re}\{\mathbf{g}^H\mathbf{B}\mathbf{f}\}+c\leq t; \\
&
\left\| 
\begin{array}{c}
(\mathbf{A}\!+\!\mathbf{C})^{\frac{1}{2}}\mathbf{f} \\
\frac{s-1}{2}
\end{array}
\right\|_2\leq \frac{s\!+\!1}{2}; \\
&\ \ \ \ \ \ 
\left\| 
\begin{array}{c}
\mathbf{D}_i^{\frac{1}{2}}\mathbf{f} \\
\frac{P_i-1}{2}
\end{array}
\right\|_2\leq \frac{P_i\!+\!1}{2},\ \ i\in\{1,\cdots,L\};
\end{align}
\end{subequations}

$(\mathsf{P}2_{SOCP})$ can be numerically solved by off-the-shelf convex programming solvers, such as CVX \cite{bib:CVX}.

Summarizing the above discussions, the problem ($\mathsf{P}0$) can be solved by a 2-BCD algorithm: updating $\mathbf{G}$ by solving ($\mathsf{P}1$) and updating $\big\{\mathbf{F}_i\big\}_{i=1}^{L}$ by solving ($\mathsf{P}2^{\prime}$) alternatively, which is summarized in Algorithm \ref{alg:two_block}. 

\begin{algorithm}
\caption{2-BCD Algorithm to Solve ($\mathsf{P}0$)}
\label{alg:two_block}
\textbf{Initialization}: Randomly generate feasible $\{\mathbf{F}_i^{(0)}\}_{i=1}^{L}$, $i\in\{1,\cdots,L\}$; Compute $\mathbf{G}^{(0)}$ using (\ref{eq:G_MMSE})\; 
\Repeat{decrease of $\mathsf{MSE}$ is small enough or predefined number of iterations is reached}
{
 With $\mathbf{G}^{(j-1)}$ fixed, solve ($\mathsf{P}2^{\prime}$) and obtain $\{\mathbf{F}_i^{(j)}\}_{i=1}^{L}$\;
 With $\{\mathbf{F}_i^{(j)}\}_{i=1}^{L}$  fixed, compute $\mathbf{G}^{(j)}$ using (\ref{eq:G_MMSE})\;
}
\end{algorithm}

\subsection{Convergence of 2-BCD Algorithm}

In this subsection we study the convergence of the above 2-BCD algorithm. Consider the optimization problem $\min\{f(\mathbf{x})|\mathbf{x}\in\mathcal{X}\}$ with $f(\cdot)$ being continuously differentiable and the feasible domain $\mathcal{X}$ being closed and nonempty. A point $\mathbf{x}_0\in\mathcal{X}$ is a \emph{stationary point} if and only if $\nabla f(\mathbf{x}_0)(\mathbf{x}-\mathbf{x}_0)\geq0$, $\forall \mathbf{x}\in\mathcal{X}$, where $\nabla f(\mathbf{x}_0)$ denotes the gradient of $f$ at $\mathbf{x}_0$. For the proposed 2-BCD algorithm, we have the following convergence conclusion. 

\begin{theorem}
\label{thm:convergence_two_block}
The objective sequence $\{\mathsf{MSE}^{(j)}\}_{j=0}^{\infty}$ generated by the 2-BCD algorithm in Algorithm \ref{alg:two_block} is monotonically decreasing. If $\mathbf{K}_i\mathbf{\Sigma}_{\mathbf{s}}\mathbf{K}_i^H\succ0$ or $\mathbf{\Sigma}_i\succ0$ for all $i\in\{1,\cdots,L\}$, the solution sequence $\big\{\{\mathbf{F}_i^{(j)}\}_{i=1}^{L},\mathbf{G}^{(j)}\big\}_{j=1}^{\infty}$ generated by the 2-BCD algorithm has limit points and each limit point of $\big\{\{\mathbf{F}_i^{(j)}\}_{i=1}^{L},\mathbf{G}^{(j)}\big\}_{j=1}^{\infty}$ is a stationary point of ($\mathit{P0}$).
\end{theorem}

\begin{proof}
Since each block update solves a minimization problem, $\mathsf{MSE}$ keeps decreasing. Let $\mathcal{X}_{i}=\big\{\mathbf{X}\in\mathbb{C}^{N_i\times J_i}\big|\Tr\{\mathbf{X}(\mathbf{K}_i\mathbf{\Sigma}_{\mathbf{s}}\mathbf{K}_i^H\!\!+\!\!\mathbf{\Sigma}_i)\mathbf{X}^H\}\leq P_i\big\}$, for $i=1,\cdots,L$ and $\mathcal{X}_{L\!+\!1}=\mathbb{C}^{M\times K}$. Under the strictly positive definiteness assumption of $\mathbf{K}_i\mathbf{\Sigma}_{\mathbf{s}}\mathbf{K}_i^H$ or $\mathbf{\Sigma}_i$, we have $\big(\mathbf{K}_i\mathbf{\Sigma}_{\mathbf{s}}\mathbf{K}_i^H\!\!+\!\!\mathbf{\Sigma}_i\big)\succ0$ and thus $\big(\mathbf{K}_i\mathbf{\Sigma}_{\mathbf{s}}\mathbf{K}_i^H\!\!+\!\!\mathbf{\Sigma}_i\big)^T\!\otimes\!\mathbf{I}_{N_i}\succ0$ for all $i\in\{1,\cdots,L\}$. This implies that the null space of $\big(\mathbf{K}_i\mathbf{\Sigma}_{\mathbf{s}}\mathbf{K}_i^H\!\!+\!\mathbf{\Sigma}_i\big)^T\!\otimes\!\mathbf{I}_{N_i}$ is $\{\mathbf{\mathbf{0}}\}$ and consequently $\mathbf{f}_i$ has to be bounded to satisfy power constraint. Therefore $\mathcal{X}_i$ is bounded for all $i\in\{1,\cdots,L\}$. Since the feasible set for each $\mathbf{F}_i$ is bounded, by Bolzano-Weierstrass theorem, there exists a convergent subsequence $\big\{\{\mathbf{F}^{(j_k)}_i\}_{i=1}^{L}\big\}_{k=1}^{\infty}$. Since $\mathbf{G}$ is updated by equation (\ref{eq:G_MMSE}) as a continuous function of $\{\mathbf{F}_i\}_{i=1}^L$, the subsequence $\{\mathbf{G}^{(j_k+1)}\}_{k=1}^{\infty}$ also converges and thus bounded. By further restricting to a subsequence of $\big\{\{\mathbf{F}_i^{(j_k+1)}\},\mathbf{G}^{(j_k+1)}\big\}_{k=1}^{\infty}$, we can obtain a convergent subsequence of $\big\{\{\mathbf{F}_i^{(j)}\}_{i=1}^{L},\mathbf{G}^{(j)}\big\}_{j=1}^{\infty}$.

Since Algorithm \ref{alg:two_block} is a two block coordinate descent procedure and the problem ($\mathsf{P}0$) has continuously differentiable objective and closed and convex feasible domain, Corollary 2 in \cite{bib:Gauss-Seidel} is valid to invoke, we conclude that any limit point of $\big\{\{\mathbf{F}_i^{(j)}\}_{i=1}^{L},\mathbf{G}^{(j)}\big\}_{j=1}^{\infty}$ is a stationary point of ($\mathsf{P}0$).
\end{proof}




\section{Multi-Block Coordinate Descent}
\label{sec:further_decoupling}

For the above 2-BCD algorithm, although we can solve the subproblem ($\mathsf{P}2$) as a standard SOCP problem, its closed-form solution is still inaccessible. The complexity for solving ($\mathsf{P}2$) can be shown to be $\mathcal{O}\Big(\sqrt{L}\big(\sum_{i=1}^{L}N_iJ_i\big)^3\Big)$, 
This implies that when the sensor network under consideration has a large number of sensors and/or antennae, the complexity for solving ($\mathsf{P}2$) can be rather daunting. This motivates us to search for more efficient ways to update sensor's beamformer.

\subsection{Further Decoupling of ($\mathit{P2}$) and Closed-Form Solution}
\label{subsec:further_decoupling}

Looking back to problem ($\mathsf{P}2$), although it has separable power constraints, its quadratic terms in its objective tangles different sensors' beamformers together and thus makes the Karush-Kuhn-Tucker(KKT) conditions of ($\mathsf{P}2$) analytically unsolvable. Here we adopt the BCD methodology to further decompose the subproblem ($\mathsf{P}2$).
Instead of optimizing all the $\mathbf{F}_i$'s in a single batch, we optimize one $\mathbf{f}_i$ at a time with the others being fixed. By introducing the notation $\mathbf{q}_i\triangleq\sum_{j=1,j\neq i}^{L}\mathbf{A}_{ij}\mathbf{f}_j$, each subproblem ($\mathsf{P}2^{\prime}_i$) of ($\mathsf{P2}^{\prime}$) is given as
\begin{subequations}
\begin{align}
\!\!\!\!\!\!\!\!\!\!(\mathsf{P}2^{\prime}_i):\underset{\mathbf{f}_i}{\min}\ &\mathbf{f}_i^H\big(\mathbf{A}_{ii}\!\!+\!\!\mathbf{C}_i\big)\mathbf{f}_i+2\Real\{\mathbf{q}_i^H\mathbf{f}_i\}\!\!-\!\!2\Real\{\mathbf{g}^H\mathbf{B}_i\mathbf{f}_i\}\\
s.t.\ & \mathbf{f}_i^H\mathbf{E}_i\mathbf{f}_i\leq P_i.\label{eq:P2prime_constraint}
\end{align}
\end{subequations}

Now our problem boils down to solving the simpler problem ($\mathsf{P}2^{\prime}_i$), for $i=1,\cdots,L$. The following theorem provides an \emph{almost} closed-form solution to ($\mathsf{P}2^{\prime}_i$). The only reason that this is not a \emph{fully} closed-form solution is because it may involve a bisection search to determine the value of a positive real number. 

\begin{theorem}\label{thm:sol_P2_i}
Assume $\mathbf{K}_i\mathbf{\Sigma}_{\mathbf{s}}\mathbf{K}_i^H\succ0$ or $\mathbf{\Sigma}_i\succ0$. Define parameters $\mathbf{M}_i$, $\mathbf{U}_i$ and $\mathbf{p}_i$ as in equations (\ref{eq:new_def}) in the appendix, $r_i$ as the rank of $\mathbf{M}_i$ and $p_{i,k}$ as the $i$-th entry of $\mathbf{p}_i$. The solution to ($\mathit{P2}_i^{\prime}$) is given as follows: \newline
\underline{$\mathsf{CASE}$ (\uppercase\expandafter{\romannumeral1})}---if either of the following two conditions holds:\newline
\indent\indent\indent $\ \ \ \ $ \expandafter{\romannumeral1}) $\exists k\in\{r_i+1,\cdots,J_iN_i\}$ such that $|p_{i,k}|\neq0$;\newline
\indent\indent\indent or \ \expandafter{\romannumeral2}) $\sum_{k\!=\!r_i\!+\!1}^{J_iN_i}|p_{i,k}|=0$ and $\sum_{k\!=\!1}^{r_i}\frac{|p_{i,k}|^2}{\lambda_{i,k}^2}>P_i$.\newline
The optimal solution to ($\it{P2}_i^{\prime}$) is given by
\begin{align}
\label{eq:opt_f_i_1}
\mathbf{f}_i^{\star}=\big(\mathbf{A}_{ii}\!\!+\!\!\mathbf{C}_i\!\!+\!\!\mu_i^{\star}\mathbf{E}_i\big)^{-1}\big(\mathbf{B}_i^H\mathbf{g}-\mathbf{q}_i\big),
\end{align}
with the positive value $\mu_i^{\star}$ being the unique solution to the equation:
$g_i(\mu_{i})=\sum_{k=1}^{J_iN_i}\frac{|p_{i,k}|^2}{(\lambda_{i,k}+\mu_i)^2}=P_i$.
An interval $\big[lbd_i, ubd_i\big]$ containing $\mu_i^{\star}$ is determined by Lemma \ref{lem:bound} which comes later.
\newline
\underline{$\mathsf{CASE}$ (\uppercase\expandafter{\romannumeral2})}---$\sum_{k\!=\!r_i\!+\!1}^{J_iN_i}|p_{i,k}|=0$ and $\sum_{k\!=\!1}^{r_i}\frac{|p_{i,k}|^2}{\lambda_{i,k}^2}\leq P_i$,\newline
The optimal solution to ($\it{P2}_i^{\prime}$) is given by
\begin{align}
\label{eq:opt_f_i_2}
\mathbf{f}_i^{\star}=\mathbf{E}_i^{-\frac{1}{2}}\Big(\mathbf{E}_i^{-\frac{1}{2}}\big(\mathbf{A}_{ii}\!\!+\!\!\mathbf{C}_i\big)\mathbf{E}_i^{-\frac{1}{2}}\Big)^{\dagger}\mathbf{E}_i^{-\frac{1}{2}}\big(\mathbf{B}_i^H\mathbf{g}-\mathbf{q}_i\big).
\end{align}
\end{theorem}
\begin{proof}
See Appendix \ref{subsec:appendix_thm_sol_P2_i}.
\end{proof}

Here we have several comments and supplementary discussions on the solution to ($\mathsf{P}2_i^{\prime}$).

\begin{comment}
\label{com:P2_i_mu0_solution_minipower}
\emph{
When $\mu_i^{\star}=0$ and $\mathbf{M}_i$ is singular, the solution to ($\mathsf{P}2_i^{\prime}$) is usually non-unique. According to the proof procedure in Appendix \ref{subsec:appendix_thm_sol_P2_i}, (\ref{eq:opt_f_i_2}) is actually the power-preserving optimal solution, which has the minimal transmission power among all optimal solutions to ($\mathsf{P}2_i^{\prime}$). 
}
\end{comment}

\begin{comment}
\label{com:P2_i_solution_structure}
\emph{
It is worth noting that the three cases discussed in the proof of Theorem \ref{thm:sol_P2_i},  $\mathsf{CASE}$(\uppercase\expandafter{\romannumeral1})-case \expandafter{\romannumeral1}), $\mathsf{CASE}$(\uppercase\expandafter{\romannumeral1})-case \expandafter{\romannumeral2}) and $\mathsf{CASE}$(\uppercase\expandafter{\romannumeral2}), are mutually exclusive events. One and only one case will occur.
}
\end{comment}

\begin{comment}
\label{com:P2_i_compare}
\emph{
The problem of minimizing MSE with respect to one separate beamformer with one power constraint is a rather standard problem that has been discussed in previous works such as \cite{bib:MIMO_multiuser, bib:MIMO_relay_multiuser, bib:sensor_network_Hamid, bib:WMMSE_QingjiangShi}. A big contribution here is that we have \emph{fully} solved this problem by clearly identifying the solution structure and writing out the almost closed-form solutions for all possible cases, whereas the previous papers have not. 
One key consideration is the case of rank deficient $\mathbf{M}_i$ for zero $\mu_i^{\star}$. Although \cite{bib:MIMO_multiuser} and \cite{bib:sensor_network_Hamid} mention that $\mu_i^{\star}$ can be zero, the solution for singular $\mathbf{M}_i$ in this case is missing. In fact when $\mathbf{M}_i$ does not have full rank and $\mu_i^{\star}$ is zero, its inverse does not exist and consequently the solutions given in \cite{bib:MIMO_multiuser, bib:MIMO_relay_multiuser, bib:sensor_network_Hamid, bib:WMMSE_QingjiangShi} do not stand any more (they all provide solutions by matrix inversion). 
It is noted that \cite{bib:sensor_network_Hamid} imposes more assumptions on the number of antennas to exclude some cases where $\mathbf{M}_i$ is rank deficient. However these assumptions undermine the generality of the system model and, still, adverse channel parameters can result in rank deficiency of $\mathbf{M}_i$. Turns out, the rank deficiency scenario is actually not rare. In fact, whenever $K<N_i$ or $M<N_i$ holds, the matrices $\mathbf{A}_{ii}$ and $\mathbf{C}_i$ are both born rank deficient. If they share common nonzero components of null space, $\mathbf{M}_i$ will be rank deficient. For example, consider the simple case where $\mathbf{K}_i=\mathbf{I}_K$, $\mathbf{\Sigma}_{\mathbf{s}}=\sigma_s^2\mathbf{I}$, $\mathbf{\Sigma}_i=\sigma_i^2\mathbf{I}$ and $\min(K,M)<N_i$. At this time $\mathbf{M}_i$ is not of full rank. Besides inappropriate channel parameters $\mathbf{H}_i$ can also generate rank deficient $\mathbf{M}_i$. Thus taking the rank deficiency of $\mathbf{M}_i$ when $\mu_{\star}=0$ into consideration is both necessary and meaningful.
}
\end{comment}

\begin{comment}
\emph{
In the special case where $K=J_i=1$, the fully closed form solution to ($\mathsf{P}2_i$) does exist! At this time, the optimal $\mu_i^{\star}$ and $\mathbf{f}_i^{\star}$ can be obtained analytically without bisection search. In this case, eigenvalue decomposition is also unnecessary. So when $K=J_i=1$, solving ($\mathsf{P}2_i$) is extremely efficient. The details can be found in \cite{bib:YangLiu_ICC}\footnote{\cite{bib:YangLiu_ICC} actually solves an approximation of problem ($\mathsf{P}2_i^{\prime}$) with scalared source , where a specific affine term of $\mathbf{f}_i$ in the objective of ($\mathsf{P}2_i^{\prime}$) is approximated by its latest value (approximation is discussed in subsection \ref{subsec:approximation} of this paper). However fully analytic solution of ($\mathsf{P}2_i$) can be obtained by following very similar lines as \cite{bib:YangLiu_ICC} without introducing approximation of $\mathbf{f}_i$.}.
}
\end{comment}


Recall that in $\mathsf{CASE}$ (\uppercase\expandafter{\romannumeral1}) of Thoerem \ref{thm:sol_P2_i}, $\mu_i^{\star}$ is obtained as the solution to $g_i(\mu_{i})=P_i$. 
This equation generally has no analytic solution. Fortunately $g_i(\mu_i)$ is strictly decreasing in $\mu_i$ and thus the equation can be efficiently solved by a bisection search. The following lemma provides an interval $[lbd_i,ubd_i]$ containing the positive $\mu_i^{\star}$, from which the bisection search to determine $\mu_i^{\star}$ can be started.

\begin{lemma}\label{lem:bound}
The positive $\mu_i^{\star}$ in ($\mathit{P2}_i^{\prime}$) (i.e. CASE (I) in Theorem \ref{thm:sol_P2_i}) has the following lower bound $lbd_i$ and  upper bound $ubd_i$:
\newline
\indent\expandafter{\romannumeral1}) For subcase \expandafter{\romannumeral1})
\begin{align}
\!\!\!\!\!\!\!\!\!\! lbd_i=\left[\frac{\|\mathbf{p}_i\|_2}{\sqrt{P_i}}-\lambda_{i,1}\right]^{+},
\ \ \  ubd_i=\frac{\|\mathbf{p}_i\|_2}{\sqrt{P_i}};\label{eq:mu_i_bound_1}
\end{align}
\indent\expandafter{\romannumeral2}) For subcase \expandafter{\romannumeral2})
\begin{align}
\ \ lbd_i=\left[\frac{\|\mathbf{p}_i\|_2}{\sqrt{P_i}}-\lambda_{i,1}\right]^{+}, \ \ \
ubd_i=\frac{\|\mathbf{p}_i\|_2}{\sqrt{P_i}}-\lambda_{i,r_i},\label{eq:mu_i_bound_2}
\end{align}
where $[x]^+=\max\{0,x\}$.
\end{lemma}
\begin{proof}
For subcase \expandafter{\romannumeral1}), by definition of $g_i(\mu_i)$ in  (\ref{eq:KKT_analysis_2}), we have
\begin{align}
\frac{\|\mathbf{p}_i\|_2^2}{(\mu_i+\lambda_{i,1})^2}=\frac{\sum_{k=1}^{J_iN_i}|p_{i,k}|^2}{(\mu_i+\lambda_{i,1})^2}&\leq g_i(\mu_i)=P_i\nonumber\\
&\!\!\!\!\!\!\!\!\!\!\leq\frac{\sum_{k=1}^{J_iN_i}|p_{i,k}|^2}{\mu_i^2}=\frac{\|\mathbf{p}_i\|_2^2}{\mu_i^2},
\end{align}
which can be equivalently written as
\begin{align}
\frac{\|\mathbf{p}_i\|_2}{\sqrt{P_i}}-\lambda_{i,1}\leq\mu_i\leq\frac{\|\mathbf{p}_i\|_2}{\sqrt{P_i}}.
\end{align}
Also notice that $\mu_i^{\star}$ should be positive; the bounds in (\ref{eq:mu_i_bound_1}) thus follow.

For subcase \expandafter{\romannumeral2}), by assumption, $\sum_{k=r_i+1}^{J_iN_i}|p_{i,k}|^2\!=\!0$. This leads to 
\begin{align}
\frac{\|\mathbf{p}_i\|_2^2}{(\mu_i+\lambda_{i,1})^2}=\frac{\sum_{k=1}^{r_i}|p_{i,k}|^2}{(\mu_i+\lambda_{i,1})^2}&\leq g_i(\mu_i)=P_i\nonumber\\
&\!\!\!\!\!\!\!\!\!\!\!\!\!\!\!\!\!\!\!\!\!\!\leq\frac{\sum_{k=1}^{r_i}|p_{i,k}|^2}{(\mu_i+\lambda_{i,r_i})^2}=\frac{\|\mathbf{p}_i\|_2^2}{(\mu_i+\lambda_{i,r_i})^2}.
\end{align}
Following the same line of derivation as in subcase \expandafter{\romannumeral1}), we obtain the bounds in (\ref{eq:mu_i_bound_2}).
\end{proof}

\begin{algorithm}
\caption{Solving the Problem ($\mathsf{P}2_i^{\prime}$)}
\label{alg:P_2_i}
\textbf{Initialization}: Perform eigenvalue decomposition $\mathbf{M}_i=\mathbf{U}_i\mathbf{\Lambda}_i\mathbf{U}_i^H$; Calculate $\mathbf{p}_i$ using (\ref{eq:p_i_definition})\;

\uIf{$\Big(\exists k\in\{r_i+1,\cdots,J_iN_i\}$ s.t. $|p_{i,k}|\neq0\Big)$ $\mathrm{or}$ $\Big(\sum_{k=r_i+1}^{J_iN_i}|p_{i,k}|^2\!=\!0$ and $\sum_{k\!=\!1}^{r_i}\!\!\!\frac{\ |p_{i,k}|^2}{\lambda_{i,k}^2}\!>\!P_i\Big)$}
{   Determine bounds $lbd_i$ and $ubd_i$ via (\ref{eq:mu_i_bound_1}) or (\ref{eq:mu_i_bound_2}) \;
    Bisection search on $\big[lbd_i,ubd_i\big]$ to determine $\mu_i^{\star}$\;
    $\mathbf{f}_i^{\star}\!=\!\big(\mathbf{A}_{ii}\!+\!\mathbf{C}_i\!\!+\!\!\mu_i^{\star}\mathbf{E}_i\big)^{-1}\big(\mathbf{B}_i^H\mathbf{g}\!-\!\mathbf{q}_i\big)$\;
}
\Else
{   
    $\mathbf{f}_i^{\star}\!=\!\mathbf{E}_i^{-\frac{1}{2}}\Big(\mathbf{E}_i^{-\frac{1}{2}}\big(\mathbf{A}_{ii}\!+\!\mathbf{C}_i\big)\mathbf{E}_i^{-\frac{1}{2}}\Big)^{\dagger}\mathbf{E}_i^{-\frac{1}{2}}\big(\mathbf{B}_i^H\mathbf{g}\!-\!\mathbf{q}_i\big)$\;
}
\end{algorithm}

Algorithm \ref{alg:P_2_i} summarizes the results obtained in Theorem \ref{thm:sol_P2_i} and Lemma \ref{lem:bound} and provides a (nearly) closed-form solution to ($\mathsf{P}2'_i$).


\subsection{Layered-BCD Algorithm}
\label{subsec:layered_block}

The above analysis of ($\mathsf{P}2_i'$), combined with ($\mathsf{P}1$), naturally leads to a nested or layered-BCD algorithm, that can be used to analytically solve the joint beamforming problem ($\mathsf{P}0$). 
The algorithm consists of two loops (two layers). The outer-loop is a two-block descent procedure alternatively optimizing $\mathbf{G}$ and $\{\mathbf{F}_i\}_{i=1}^{L}$, and the inner-loop further decomposes the optimization of $\{\mathbf{F}_i\}_{i=1}^{L}$ into an $L$-block descent procedure operated in an iterative round robin fashion. Algorithm \ref{alg:layered_block} outlines the overall procedure. As will be seen in the next, this layered-BCD has strong convergence property. 

\begin{algorithm}
\caption{Layered-BCD Algorithm to Solve ($\mathsf{P}0$)}
\label{alg:layered_block}
\textbf{Initialization}: Randomly generate feasible $\{\mathbf{F}_i^{(0)}\}_{i=1}^{L}$ \;
        Obtain $\mathbf{G}^{(0)}$ by (\ref{eq:G_MMSE})\; 
\Repeat{decrease of $\mathsf{MSE}$ is sufficiently small or predefined number of iterations is reached}
{
    \Repeat{decrease of $\mathsf{MSE}$ is sufficiently small}	
	{
		\For{$i=1;\ i<=L;\ i++$}
		{
		  Given $\mathbf{G}$ and $\{\mathbf{F}_j\}_{j\neq i}$, update $\mathbf{F}_i$ by Theorem \ref{thm:sol_P2_i}\;
		}
	}
    Given $\big\{\mathbf{F}_i\big\}_{i=1}^{L}$, update $\mathbf{G}^{}$ via (\ref{eq:G_MMSE}) \;
}
\end{algorithm}

\begin{theorem}
\label{thm:convergence_layered_block}
Assume that $\mathbf{K}_i\mathbf{\Sigma}_{\mathbf{s}}\mathbf{K}_i^H\succ0$ or $\mathbf{\Sigma}_i\succ0$, $\forall i\in\{1,\cdots,L\}$. The objective sequence $\{\mathsf{MSE}^{(j)}\}_{j=0}^{\infty}$ generated by Algorithm \ref{alg:layered_block} is monotonically decreasing. The solution sequence $\big\{\{\mathbf{F}_i^{(j)}\}_{i=1}^{L},\mathbf{G}^{(j)}\big\}_{j=1}^{\infty}$ generated by Algorithm \ref{alg:layered_block} has limit points, and each limit point 
is a stationary point of ($\mathit{P}0$).
\end{theorem}

\begin{proof}
The proof of the monotonicity of $\{\mathsf{MSE}^{(j)}\}_{j=0}^{\infty}$ and the existence of limit points for the solution sequence follows the same lines as those of Theorem \ref{thm:convergence_two_block}.

From Theorem \ref{thm:P_2_convexity}, given $\mathbf{G}$, the objective function $\mathsf{MSE}\big(\{\mathbf{F}_i\}_{i=1}^{L}\big|\mathbf{G}\big)$ of Problem (\ref{eq:opt_prob_P2_prime}) is convex (and therefore, of course, pseudoconvex) with respect to $\{\mathbf{f}_i\}_{i=1}^{L}$. Since the objective $\mathsf{MSE}\big(\{\mathbf{F}_i\}_{i=1}^{L}\big|\mathbf{G}\big)$ in ($\mathsf{P}2$) is continuous and the feasible domain of $\{\mathbf{F}_i\}_{i=1}^{L}$ is bounded, there exists some feasible point $\{\bar{\mathbf{F}}_i\}_{i=1}^{L}$ making the level set $\Big\{\{\mathbf{F}_i\}_{i=1}^{L}\in\mathbb{C}^{J_1N_1\times 1}\times\cdots\times\mathbb{C}^{J_LN_L\times 1}\big|\,\mathsf{MSE}\big(\{\mathbf{F}_i\}_{i=1}^{L}\big|\mathbf{G}\big)\leq \mathsf{MSE}\big(\{\bar{\mathbf{F}}_i\}_{i=1}^{L}\big|\mathbf{G}\big)\Big\}$ closed and bounded. Thus Proposition 6 in \cite{bib:Gauss-Seidel} is valid to invoke. For a given $\mathbf{G}$  at any step of outer-loop, the inner loop generates limit point(s) converging to a stationary point of the problem ($\mathsf{P}2$). Since ($\mathsf{P}2$) is a convex problem, any stationary point is actually an optimal solution \cite{bib:nonlinear_programming_Bertsekas}. Therefore the subproblem ($\mathsf{P}2$) is actually globally solved. By Theorem \ref{thm:convergence_two_block}, each limit point of solution sequence is a stationary point of the original problem ($\mathsf{P}0$).
\end{proof}

Although the convergence analysis in Theorem \ref{thm:convergence_layered_block} states that the layered-BCD algorithm guarantees convergence, it requires the inner-loop to iterate numerous times to converge sufficiently. In fact if each inner loop is performed with a small number of iterations, the layered BCD algorithm becomes a specialized essentially cyclic BCD algorithm, which will be discussed in next subsection.

\subsection{Essentially Cyclic $(L+1)$-BCD Algorithm}
\label{subsec:EssentiallyCyclci}

In this subsection, we propose an ($L+1$)-BCD algorithm, where in each update the linear FC receiver or one single beamformer is updated efficiently by equation (\ref{eq:G_MMSE}) or Theorem \ref{thm:sol_P2_i} respectively. Compared to the 2-BCD algorithm, the block updating rule for multiple block coordinate descent method can have various patterns. Here we adopt a very general updating manner called \emph{essentially cyclic rule} \cite{bib:block_coordinate_decent_algorithm}. For essentially cyclic update rule, there exists a positive integer $T$, which is called period, such that each block of variables is updated at least once within any consecutive $T$ updates. The classical Gauss-Seidel method is actually a special case of essentially cyclic rule with its period $T$ being exactly the number of blocks of variables.

For the convergence of essentially cyclic BCD algorithm, when the whole solution sequence converges, the limit of the solution sequence is stationary. In fact, assume that the sequence $\big\{\{\mathbf{F}_i^{(j)}\}_{i=1}^{L},\mathbf{G}^{(j)}\big\}_{j=1}^{\infty}$ converges to the limit point $\bar{\mathbf{X}}\triangleq\big\{\{\bar{\mathbf{F}}_i\}_{i=1}^{L},\bar{\mathbf{G}}\big\}$. Denote $\mathbf{X}=\big\{\{\mathbf{F}_i\}_{i=1}^L,\mathbf{G}\big\}$ and $\mathbf{X}_i$ as the $i$-th block of $\mathbf{X}$, which can be $\mathbf{G}$ or $\mathbf{F}_j$, $\forall j\in\{1,\cdots,L\}$, and $\mathbf{X}_{\bar{i}}$ as the variables other than $\mathbf{X}_i$, i.e. $\mathbf{X}_{\bar{i}}=\{\mathbf{X}\}\backslash\{\mathbf{X}_i\}$. Since $\mathbf{X}_i^{(j+1)}$ minimizes $\mathsf{MSE}$ with given $\{\mathbf{X}_{\bar{i}}^{(j)}\}$, as optimality conditions, we have $\Tr\{\mathbf{\nabla}_{\mathbf{X}_i}\mathsf{MSE}(\mathbf{X}_i^{(j+1)},\mathbf{X}_{\bar{i}}^{(j)})^T\big(\mathbf{X}_i\!-\!\mathbf{X}_i^{(j+1)}\big)\}\geq0$ for any feasible $\mathbf{X}_i$. Since $\{\mathbf{X}_{\bar{i}}^{(j)}\}\rightarrow\bar{\mathbf{X}}_{\bar{i}}$, $\mathbf{X}_i^{(j+1)}\rightarrow\bar{\mathbf{X}}_i$ and $\mathsf{MSE}$ is continuously differentiable, we have $\Tr\{\mathbf{\nabla}_{\mathbf{X}_i}\mathsf{MSE}(\bar{\mathbf{X}})^T\big(\mathbf{X}\!-\!\bar{\mathbf{X}}_i\big)\}\geq0$ for any feasible $\mathbf{X}_i$, $\forall i\in\{1,\cdots,L+1\}$. By summing up all $L+1$ variable blocks, we obtain $\Tr\{\mathbf{\nabla}_{\mathbf{X}}\mathsf{MSE}(\bar{\mathbf{X}})^T\big(\mathbf{X}\!-\!\bar{\mathbf{X}}\big)\}\geq0$ for any feasible $\mathbf{X}$. This suggests that the convergent limit point $\big\{\{\bar{\mathbf{F}}_i\}_{i=1}^{L},\bar{\mathbf{G}}\big\}$ is actually a stationary point of ($\mathsf{P}0$).

However the assumption that the whole solution sequence converges is actually a very strong assumption and cannot be theoretically proved, although extensive numerical results show that this fact seem always hold in practice for our problem. 

For rigorous proof of the convergence to stationary points of BCD algorithms, one usually requires uniqueness of solutions for each block update, as the analysis performed in \cite{bib:Gauss-Seidel, bib:nonlinear_programming_Bertsekas, bib:block_coordinate_decent_algorithm}. Without the uniqueness assumptions, convergence to stationary points is not guaranteed and a counter example has been reported in \cite{bib:BCD_counterExample}, where the solution sequence is always far away from stationary points. In retrospect to Theorem \ref{thm:sol_P2_i}, specific parameter settings ($\mathsf{CASE}$(\uppercase\expandafter{\romannumeral2}) with singular $(\mathbf{A}_{ii}\!+\!\mathbf{C}_i)$ and zero $\mu_i^{\star}$) will result in infinitely many optimal solutions to ($\mathsf{P}2_i^{\prime}$). To overcome this difficulty, we adopt proximal method (Exercise 2.7.1 in \cite{bib:nonlinear_programming_Bertsekas}), which locally modifies the ($\mathsf{P}2_i^{\prime}$) by imposing a squared norm and guarantees that each block update is uniquely solved. 

Specifically, to update the $i$-th beamformer, we consult to the proximal version objective $\mathsf{MSE}\big(\mathbf{f}_i\big|\{\!\mathbf{f}_j\!\}_{j\neq i},\mathbf{g}\big)\!+\!\kappa\|\mathbf{f}_i\!-\!\hat{\mathbf{f}}_i\|_2^2$ of ($\mathsf{P}2_i^{\prime}$) with $\hat{\mathbf{f}}_i$ being the latest value of $\mathbf{f}_i$ until the current update and $\kappa$ being any positive real constant. Thus the problem updating the $i$-th sensor's beamformer is equivalent to ($\mathsf{P}4_i$) as follows
\begin{align}
\!\!\!\!\!\!(\mathsf{P}4^{}_i)\!:\!\underset{\mathbf{f}_i}{\min}\ &\mathbf{f}_i^H\big(\mathbf{A}_{ii}\!\!+\!\!\mathbf{C}_i\!\!+\!\!\kappa\mathbf{I}_{N_iJ_i}\big)\mathbf{f}_i\!\!+\!\!2\Real\big\{\big(\mathbf{q}_i^H\!\!-\!\!\mathbf{g}^H\mathbf{B}_i\!\!-\!\!\kappa\hat{\mathbf{f}}_i^H\big)\mathbf{f}_i\big\}\nonumber\\
s.t.\ & \mathbf{f}_i^H\mathbf{E}_i\mathbf{f}_i\leq P_i.\label{eq:P2prime_constraint}
\end{align}
As shown by the following theorem, the proximal version of any essentially cyclic ($L+1$)-BCD algorithm guarantees monotonic decreasing of objective and stationary-point-achieving convergence of the solution sequence. 

\begin{theorem}
\label{thm:converge_essential_cyclic}
Assume that $\mathbf{K}_i\mathbf{\Sigma}_{\mathbf{s}}\mathbf{K}_i^H\succ0$ or $\mathbf{\Sigma}_i\succ0$, $\forall i\in\{1,\cdots,L\}$. By updating $\mathbf{G}$ and $\mathbf{F}_i$ by solving ($\mathit{P1}$) and ($\mathit{P4}_i$) respectively, any essentially cyclic ($\mathit{L\!+\!1}$)-BCD algorithm generates monotonically decreasing $\mathsf{MSE}$ sequence and the solution sequence has limit points with each limit point being a stationary point of the original problem ($\mathit{P0}$). 
\end{theorem}
\begin{proof}
See Appendix \ref{subsec:thm_converge_essential_cyclic}.
\end{proof}

Note the solution to ($\mathsf{P}4_i$) can be easily obtained by Theorem 3 with the terms $\big(\mathbf{A}_{ii}\!+\!\mathbf{C}_i\big)$ and $\big(\mathbf{B}_i^H\mathbf{g}\!-\!\mathbf{q}_i\big)$ being replaced by $\big(\mathbf{A}_{ii}\!+\!\mathbf{C}_i\!+\!\kappa\mathbf{I}_{N_iJ_i}\big)$ and $\big(\mathbf{B}_i^H\mathbf{g}\!-\!\mathbf{q}_i\!+\!\kappa\hat{\mathbf{f}}_i\big)$ respectively and no additional complexity is required.

Recall the layered-BCD algorithm discussed in previous subsection, when the inner-loop is performed by small number of iterations, it actually reduces to a specialized essentially cyclic BCD algorithm. One special case is the iterative algorithm proposed in \cite{bib:sensor_network_Hamid} whose inner-loop updates each beamformer for once. According to the above theorem, by updating each beamformer with the proximal method, the convergence to stationary points can be guaranteed.

One drawback of the above proximal update is its slow convergence rate, as will be shown in Section \ref{sec:numerical results}. However this shortcoming can be well compensated by the following acceleration scheme in the next subsection.

\subsection{Acceleration by Approximation}
\label{subsec:approximation}

The aforementioned ($L\!+\!1$)-BCD algorithm can be accelerated by introducing approximation when updating single beamformer $\mathbf{F}_i$ in ($\mathsf{P}2_i^{\prime}$). In addition to setting the $\{\mathbf{F}_j\}_{j\neq i}$ as known and fixed, we assume that the term $\mathbf{A}_{ii}\mathbf{f}_i$ is also known by leveraging the value of $\mathbf{f}_i$ in the previous updates. In other words, we define $\hat{\mathbf{q}}_i=\sum_{j=1,j\neq i}^{L}\mathbf{A}_{ij}\mathbf{f}_j\!+\!\mathbf{A}_{ii}\hat{\mathbf{f}}_i\!=\!\mathbf{q}_i\!+\!\mathbf{A}_{ii}\hat{\mathbf{f}}_i$ with $\hat{\mathbf{f}}_i$ being the latest value of $\mathbf{f}_i$. Thus to update $\mathbf{f}_i$ we solve the approximate version ($\mathsf{P}5_i$) of ($\mathsf{P}2_i^{\prime}$) as follows
\begin{subequations}
\begin{align}
\!\!\!\!\!\!\!\!\!\!(\mathsf{P}5_i):\underset{\mathbf{f}_i}{\min}\ &\mathbf{f}_i^H\mathbf{C}_i\mathbf{f}_i+2\Real\{\hat{\mathbf{q}}_i^H\mathbf{f}_i\}\!\!-\!\!2\Real\{\mathbf{g}^H\mathbf{B}_i\mathbf{f}_i\}\\
s.t.\ & \mathbf{f}_i^H\mathbf{E}_i\mathbf{f}_i\leq P_i.\label{eq:P2prime_constraint}
\end{align}
\end{subequations}
The problem ($\mathsf{P}5_i$) can still be efficiently solved by Theorem \ref{thm:sol_P2_i}. Interestingly enough, this approximation can significantly improve the convergence rate of the cyclic-BCD procedure!

Actually similar idea appears in \cite{bib:MIMO_multiuser}, where the precoders of multiusers is updated in a cyclic manner. In Implementation 2 (Table II) of \cite{bib:MIMO_multiuser}, with others being fixed, one separate precoder is updated by minimizing the total $\mathsf{MSE}$ function with some terms of the to-be-updated precoder approximated by previous values. As reflected by the extensive numerical results in \cite{bib:MIMO_multiuser}, this approximated BCD implementation has surprisingly faster convergence compared to the original one (Implementation I in Table I) in \cite{bib:MIMO_multiuser}.  

The surprisingly fast convergence of the approximate update inspires us the idea that it can become perfect complement of the aforementioned proximal update. In implementation, ($L\!+\!1$)-BCD algorithms can be performed in an approximate-proximal manner---in the first few outer-loop iterations we run the approximate update and then convert to proximal update in the subsequent updates. This approximate-proximal combination exhibits fast convergence and also guarantees stationary-points-achieving convergence as shown previously.


\section{Numerical Results}
\label{sec:numerical results}

In this section, numerical results are presented to verify and compare the performance of the proposed algorithms.  

In the following experiments, a wireless sensor network with $L=3$ sensors is considered. The antenna numbers of the sensors and the fusion center are set as $N_1=3, N_2=4, N_3=5$ and $M=4$ respectively. All observation matrices $K_i$ are set as identity matrices. The source signal $\mathbf{s}$ has dimension $K=3$ with zero mean, unit-power and uncorrelated components. 
The observation noise at each sensor is colored and has covariance matrix $\mathbf{\Sigma}_i=\sigma_i^2\mathbf{\Sigma}_{0,i}$, $i\in\{1,\cdots,L\}$, where the $J_i\times J_i$ matrix $\mathbf{\Sigma}_{0,i}$ has the Toeplitz structure with its ($j,k$)-th element $[\mathbf{\Sigma}_{0,i}]_{j,k}\!=\!\rho^{|k\!-\!j|}$. The parameter $\rho$ is set as $0.5$ for all sensors in our test. The transmission power and observation noise at each sensor are set as $P_1=2$, $P_2=2$, $P_3=3$, $\sigma_1^{-2}=6$dB, $\sigma_2^{-2}=7$dB and $\sigma_3^{-2}=8$dB, respectively.

In the test of each algorithm, channel noise level increases from $\mathsf{SNR}_0=0$dB to $18$dB. For one specific channel noise level, $500$ channel realizations $\{\mathbf{H}_1,\mathbf{H}_2,\mathbf{H}_3\}$ are randomly generated with each matrix entry following standard complex circular Gaussian distribution $\mathcal{CN}(0,1)$. The mean square error averaged over all $500$ random channel realizations are evaluated as a function of the number of (outer-loop) iterations and the channel SNR.

2-BCD algorithm is implemented by utilizing CVX(with SDPT3 solver) to solve its subproblem ($\mathsf{P}2$). For the essentially cyclic ($L\!+\!1$)-BCD algorithm, here we test two special cases: i) the layered BCD algorithm with finite inner-loop iterations, where the inner-loop cyclically updates each beamformer for two times; ii) ($L\!+\!1$)B-FG algorithm, where beamformers are cyclically updated with each $\mathbf{F}_i$'s update followed by the calibration of $\mathbf{G}$. That means the variables are updated in an order of $\mathbf{F}_1,\mathbf{G},\mathbf{F}_2,\mathbf{G},\cdots$. In one outer-loop it updates each $\mathbf{F}_i$ once and $\mathbf{G}$ for $L$ times. The performance of these two cases are presented in Figure \ref{fig:MSE_alg1_vs_alg2_layered} and \ref{fig:MSE_alg1_vs_alg2} respectively. The 2-BCD algorithm is plotted in each figure to serve as a benchmark. On average, the layered-BCD algorithm with finite inner-loop iteration and the ($L\!+\!1$)B-FG algorithm converges in 30-40 outer-loop iterations to the identical MSE as that of the 2-BCD algorithm.  

\begin{figure}
\centering
\includegraphics[height=2.7in,width=3.7in]{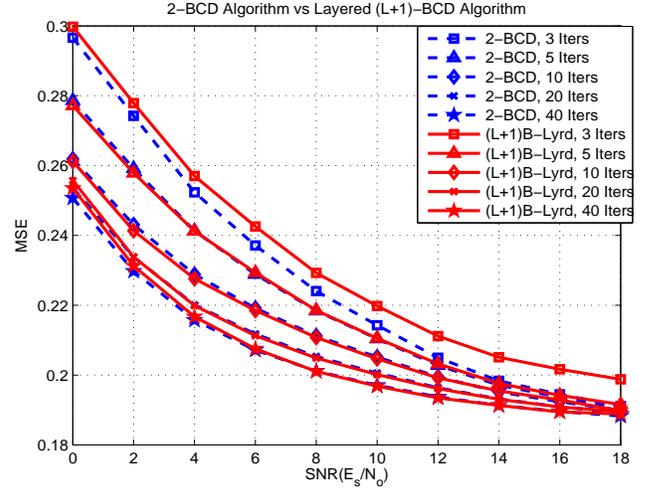}
\caption{$\mathsf{MSE}$ Performance of 2-BCD  v.s. Layered ($L\!+\!1$)-BCD (with 2 inner-loop iterations) Algorithms}
\label{fig:MSE_alg1_vs_alg2_layered} 
\end{figure}

\begin{figure}
\centering
\includegraphics[height=2.7in,width=3.7in]{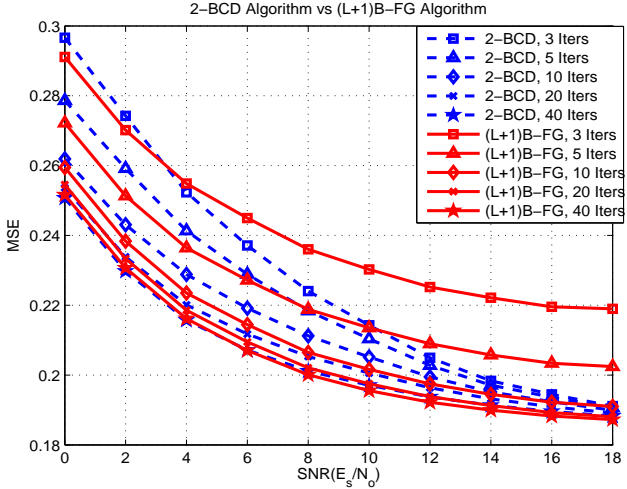}
\caption{$\mathsf{MSE}$ Performance of 2-BCD v.s.($L\!+\!1$)B-FG Algorithms}
\label{fig:MSE_alg1_vs_alg2} 
\end{figure}

The approximate and proximal version (with $\kappa=1$) of ($L\!+\!1$)B-FG algorithm are also tested and presented in Figure \ref{fig:MSE_alg1_vs_alg2_approx} and \ref{fig:MSE_alg1_vs_alg2_prox} respectively. As shown in Figure \ref{fig:MSE_alg1_vs_alg2_approx}, the performance of approximate method is surprisingly fast and exhibits excellent performance within only 3 to 5 outer-loop iterations.
Comparatively the proximal method, although whose convergence to stationary points can be proved, exhibits a much slower convergence than other algorithms, as shown in Figure \ref{fig:MSE_alg1_vs_alg2_prox}.

\begin{figure}
\centering
\includegraphics[height=2.7in,width=3.7in]{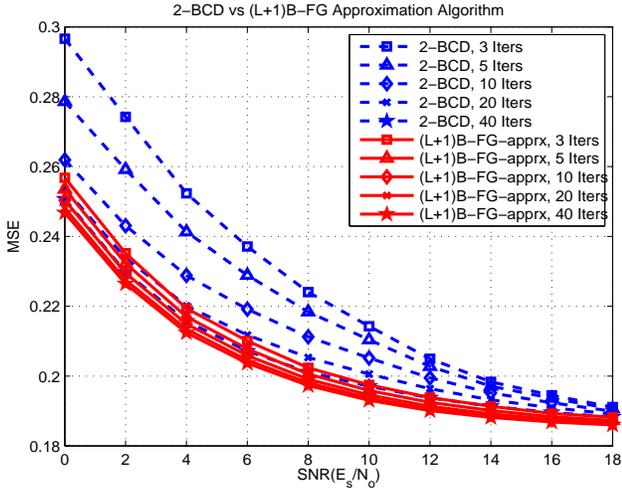}
\caption{$\mathsf{MSE}$ Performance of 2-BCD v.s. Approximate ($L\!+\!1$)B-FG}
\label{fig:MSE_alg1_vs_alg2_approx} 
\end{figure}

\begin{figure}
\centering
\includegraphics[height=2.7in,width=3.7in]{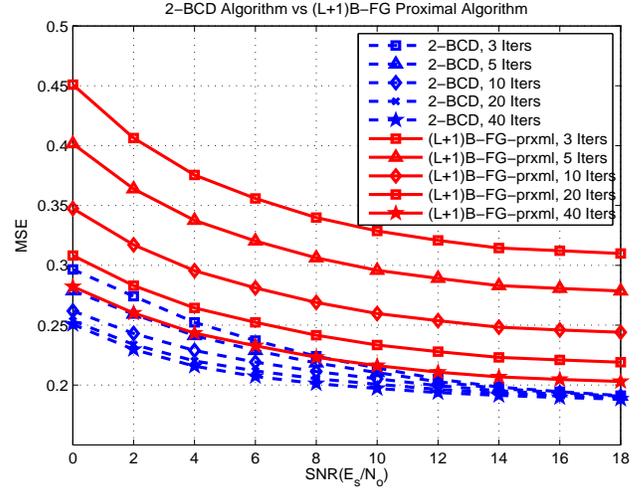}
\caption{$\mathsf{MSE}$ Performance of 2-BCD v.s. Proximal ($L\!+\!1$)B-FG}
\label{fig:MSE_alg1_vs_alg2_prox} 
\end{figure}

In Figure \ref{fig:MSE_alg1_vs_alg2_ApproxProx} the approximate-proximal version of ($L\!+\!1$)B-FG is tested. Here in the first 10 outer-loop iterations, approximate version of ($L\!+\!1$)B-FG is performed and after that the proximal method is used. As shown in the figure, this combination scheme inherits the fast convergence rate of approximate method and, as proved previously, guarantees convergence to stationary points.

\begin{figure}
\centering
\includegraphics[height=2.7in,width=3.7in]{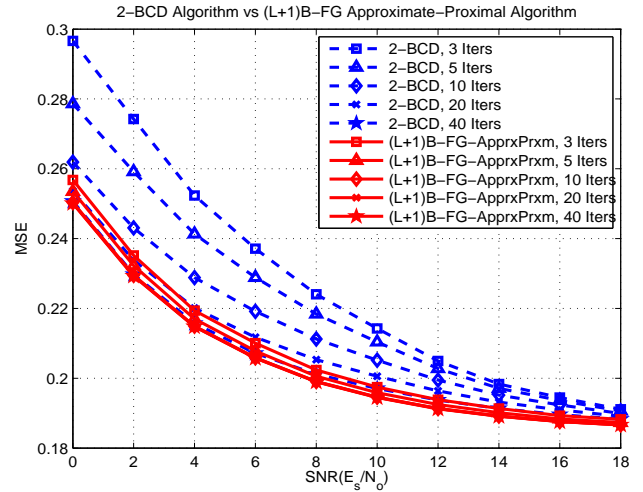}
\caption{$\mathsf{MSE}$ Performance of 2-BCD v.s. Approximate-Proximal ($L\!+\!1$)B-FG}
\label{fig:MSE_alg1_vs_alg2_ApproxProx} 
\end{figure}

Next, we take a close look at the convergence behaviors of these algorithms. We set $\mathsf{SNR}_0=2$dB and fix the channel by a randomly-generated realization. We randomly generate 10 feasible initial points. We run 2-BCD, ($L\!+\!1$)B-FG, proximal ($L\!+\!1$)B-FG and approximate-proximal ($L\!+\!1$)B-FG algorithms from these 10 random initial points and represent the resultant $\mathsf{MSE}$ itineraries  in Figures \ref{fig:itinerary_alg1_vs_alg2}-\ref{fig:itinerary_alg1_vs_alg2_approx_prox}. These plots clearly demonstrate that these algorithms are insensitive to initial points and exhibit rather stable converged MSE from different startings. As shown in the figures, different algorithms with random initials finally converge to identical MSE value with different
convergence rates. Proximal method has an obviously slower convergence and the approximate-proximal method exhibits fast convergence in the first 3 outer-loop iterations, which coincides with the observations presented in previous figures.


\begin{figure}
\centerline{
\includegraphics[height=2.7in,width=3.7in]{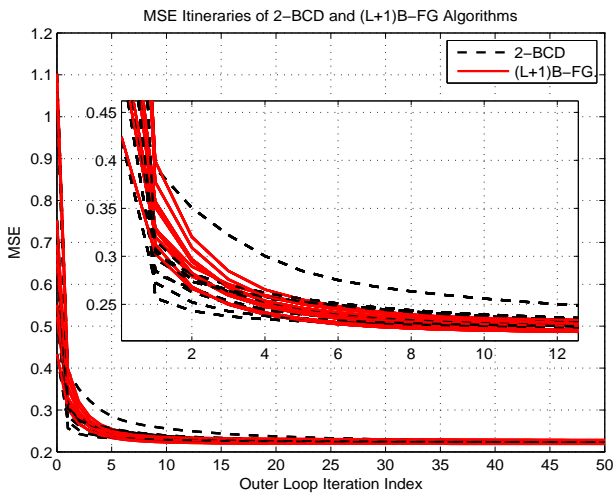}
}
\caption{MSE Itineraries of 2-BCD v.s. ($L\!+\!1$)B-FG Algorithm}
\label{fig:itinerary_alg1_vs_alg2} 
\end{figure}

\begin{figure}
\centerline{
\includegraphics[height=2.7in,width=3.7in]{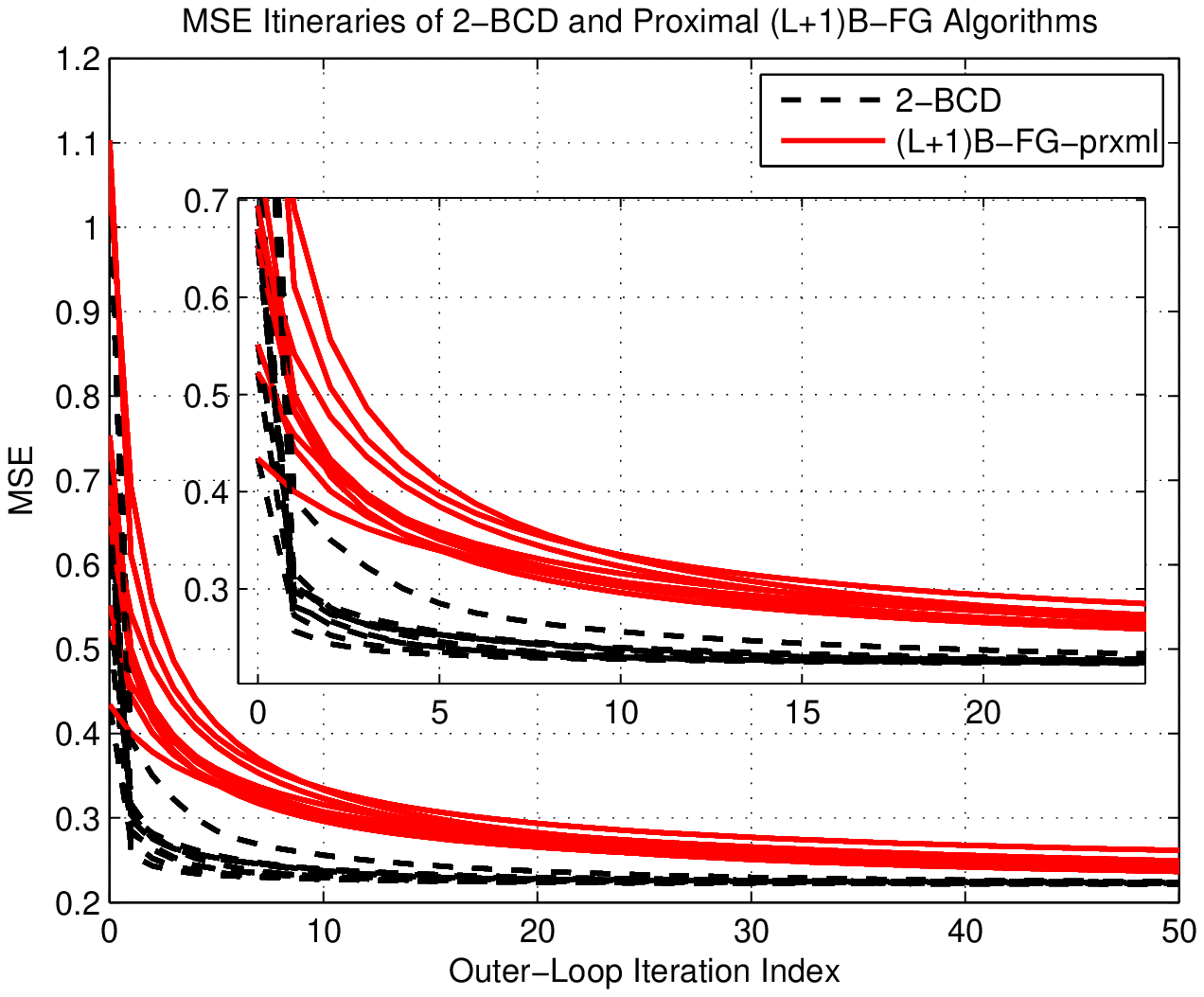}
}
\caption{MSE Itineraries of 2-BCD v.s. Proximal ($L\!+\!1$)B-FG Algorithm}
\label{fig:itinerary_alg1_vs_alg2_proximal} 
\end{figure}

\begin{figure}
\centering
\includegraphics[height=2.7in,width=3.7in]{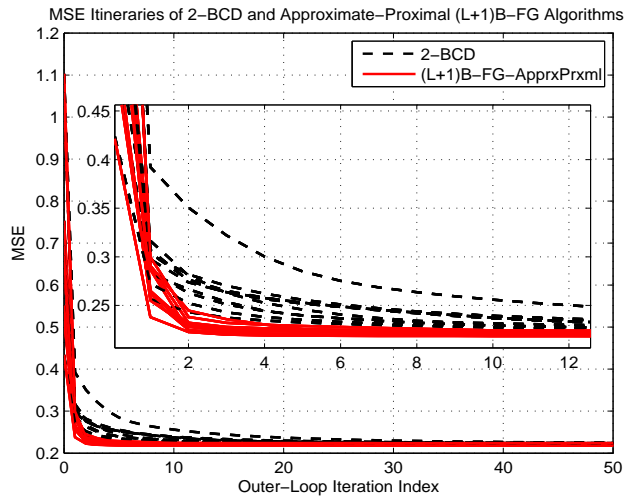}
\caption{MSE Itineraries of 2-BCD  v.s. Approximate-Proximal ($L\!+\!1$)B-FG Algorithm}
\label{fig:itinerary_alg1_vs_alg2_approx_prox} 
\end{figure}

We present in Table \ref{tab:runningtime} the average MATLAB running time for different algorithms (running in a regular laptop). For simplicity, we focus on homogeneous sensor network, where each sensor has the same number of antennae and $J_i=N_i$. Different values of $K$ (size of the source vector) and $L$ (number of sensors) are tested to take into account different problem sizes. 
The algorithms are run multiple times and the average MATLAB running time per outer-loop iteration is recorded. For the 2-BCD algorithm, CVX is utilized to solve its subproblem and the solver SDPT3 is chosen. In Table \ref{tab:runningtime}, the average running time of 2-BCD, ($L\!+\!1$)B-FG and layered ($L\!+\!1$)-BCD algorithm is presented. Note that the approximate, proximal and approximate-proximal ($L\!+\!1$)B-FG algorithms have the same complexity as that of ($L\!+\!1$)B-FG algorithm. The analytic solutions obtained in Theorem \ref{thm:sol_P2_i} entitles the essentially cyclic ($L\!+\!1$)-BCD algorithm and its variants high efficiency for implementation. However it should be pointed out that the 2-BCD method still has great significance in decentralized optimization for our problem. As proved in Theorem \ref{thm:P_2_convexity}, its subproblem ($\mathsf{P}2$) is convex. In fact, by taking advantage of this key property and utilizing multiplier method, the problem ($\mathsf{P}0$) can be solved under the 2-BCD framework, where ($\mathsf{P}2$) is solved in a highly distributed manner with each sensor updating its own beamformer.

\begin{table}
\caption{\small{MATLAB Running Time Per (Outer-Loop) Iteration}}
\centering
\begin{tabular}{|c|c|c|c|c|c|}
\hline
\backslashbox{Dim.}{$L$}& Algorithms &$L=2$&$L=4$&$L=6$&$L=8$ \\
\hline
$K=1$ & 2-BCD & 0.2167s & 0.2490 & 0.2987s & 0.3500s \\
\cline{2-6}
$M=3$ &($L\!+\!1$)B-FG & 0.0026s & 0.0066s & 0.0120s & 0.0189s \\
\cline{2-6}
$N_i=3$ & Lay. BCD & 0.0031s & 0.0094s & 0.0181s & 0.0301s \\
\hline
$K=3$ & 2-BCD & 0.2432s & 0.3068s & 0.3636s & 0.4285s \\ \cline{2-6}
$M=3$ &($L\!+\!1$)B-FG & 0.0056s & 0.0159s & 0.0328s & 0.0560s \\ \cline{2-6}
$N_i=3$ & Lay. BCD  & 0.0087s & 0.0241s & 0.0493s & 0.0839s \\
\hline
$K=6$ & 2-BCD & 0.2529s & 0.3786s & 0.5861s & 0.7526s \\ \cline{2-6}
$M=6$ &($L\!+\!1$)B-FG & 0.0075s & 0.0203s & 0.0397s & 0.0664s \\ \cline{2-6}
$N_i=6$ & Lay. BCD & 0.0116s & 0.0319s & 0.0622s & 0.1031s \\
\hline
$K=9$ & 2-BCD & 0.4352s & 0.7956s & 1.1401s & 1.9593s \\ \cline{2-6}
$M=9$ &($L\!+\!1$)B-FG & 0.0120s & 0.0302s & 0.0557s & 0.0902s \\ \cline{2-6}
$N_i=9$ & Lay. BCD & 0.0205s & 0.0514s & 0.0928s & 0.1467s \\
\hline
\end{tabular}
Notes: (i) layered-BCD is run with 2 inner-loop iterations.\\
\qquad\qquad\qquad (ii) SDPT3 solver of CVX is chosen to implement 2-BCD.
\label{tab:runningtime}
\end{table}

\section{Conclusion}
\label{sec:conclusion}

In this paper we study the joint transceiver design problem for the wireless sensor network under the $\mathsf{MSE}$ criterion. Due to the nonconvexity of the original problem, block coordinate descent methods are adopted. A two-block coordinate descent method is first proposed, which decomposes the original problem into two subproblems and alternatively optimizes the linear postcoder and the linear precoders jointly. This 2-BCD algorithm guarantees convergence (of its solution limit points) to stationary points. We also completely solve the one single beamformer's optimization problem with one power constraint. This conclusion gives birth to highly efficient multiple block coordinate descent methods. We prove the fact that updating the separate beamformer or the linear receiver in any essentially cyclic rule with proximal method can guarantee the convergence to stationary points. Moreover combining approximation with the proximal method significantly improves the convergence rate  while maintaining its strong convergence and high efficiency. Extensive numerical results are provided to verify our findings. 


\appendix

\subsection{Proof of Theorem \ref{thm:sol_P2_i}}
\label{subsec:appendix_thm_sol_P2_i}

\begin{proof}
The assumption implies $\big(\mathbf{K}_i\mathbf{\Sigma}_{\mathbf{s}}\mathbf{K}_i^H\!\!+\!\!\mathbf{\Sigma}_i\big)\succ0$. Therefore $\mathbf{E}_i=\big(\mathbf{K}_i\mathbf{\Sigma}_{\mathbf{s}}\mathbf{K}_i^H\!\!+\!\!\mathbf{\Sigma}_i\big)\!\otimes\!\mathbf{I}_{N_i}\succ0$. We introduce the following notations 
\begin{subequations}
\label{eq:new_def}
\begin{align}
\widetilde{\mathbf{f}}&\triangleq\mathbf{E}_i^{\frac{1}{2}}\mathbf{f};\label{eq:fi_redef}\\
\mathbf{M}_i&\triangleq\mathbf{E}_i^{-\frac{1}{2}}\big(\mathbf{A}_{ii}\!+\!\mathbf{C}_i\big)\mathbf{E}_i^{-\frac{1}{2}}=\mathbf{U}_i\mathbf{\Lambda}_i\mathbf{U}_i^H; \label{eq:M_eigendecomp}\\
\mathbf{b}_i&\triangleq \mathbf{E}_i^{-\frac{1}{2}}\big(\mathbf{B}_i^H\mathbf{g}-\mathbf{q}_i\big);\label{eq:b_def}\\
\mathbf{p}_i&\triangleq\mathbf{U}_i^H\mathbf{b}_i;\label{eq:p_i_definition}
\end{align}
\end{subequations}
where the $j$-th column $\mathbf{u}_{i,j}$ of $\mathbf{U}_i$ is the eigenvector associated with the eigenvalue $\lambda_{i,j}\triangleq[\mathbf{\Lambda}_i]_{j,j}$. Without loss of generality, we assume that the eigenvalues of $\mathbf{M}_i$ are arranged in a decreasing order and that $\mathbf{M}_i$ has rank $r_i$, $r_i\leq J_iN_i$. In other words $\lambda_{i,1}\geq\cdots\geq\lambda_{i,r_i}>\lambda_{i,r_i+1}=\cdots=\lambda_{i,J_iN_i}=0$.

Then the problem ($\mathsf{P}2_i^{\prime}$) is rewritten as
\begin{subequations}
\begin{align}
(\mathsf{P}3_i):\underset{\widetilde{\mathbf{f}}_i}{\min}\ &\widetilde{\mathbf{f}}_i^H\mathbf{M}_i\widetilde{\mathbf{f}}_i-2\Real\{\mathbf{b}_i^H\widetilde{\mathbf{f}}_i\}, \\
s.t.\ & \|\widetilde{\mathbf{f}}_i\|^2_2\leq P_i.
\end{align}
\end{subequations} 

Since $\mathbf{M}_i$ is positive semidefinite, ($\mathsf{P}3_i$) is convex and is obviously strictly feasible. Thus solving ($\mathsf{P}3_i$) is equivalent to solving its KKT conditions: 
\begin{subequations}
\begin{align}
\big(\mathbf{M}_i+\mu_i\mathbf{I}\big)\widetilde{\mathbf{f}}_i
&=\mathbf{b}_i;\label{eq:P2_i_KKT_a}\\
\|\widetilde{\mathbf{f}}_i\|^2_2&\leq P_i;\label{eq:P2_i_KKT_b}\\
\mu_i\big(\|\widetilde{\mathbf{f}}_i\|^2_2-P_i\big)&=0;\label{eq:P2_i_KKT_c}\\
\mu_i &\geq0. \label{eq:P2_i_KKT_d}
\end{align}\label{eq:P2_i_KKT}
\end{subequations}
The Lagrangian multiplier $\mu_i$ should be either positive or zero, our next discussion focuses on identifying the positivity of $\mu_i$.

Assume that $\mu_i>0$, then $\big(\mathbf{M}_i+\mu_i\mathbf{I}\big)$ is strictly positive definite and thus invertible. Consequently $\widetilde{\mathbf{f}}_i=\big(\mathbf{M}_i+\mu_i\mathbf{I}\big)^{-1}\mathbf{b}_i$. By the slackness condition (\ref{eq:P2_i_KKT_c}), the power constraint (\ref{eq:P2_i_KKT_b}) should be active. Plugging $\widetilde{\mathbf{f}}_i$ into (\ref{eq:P2_i_KKT_b}) and noting the eigenvalue decomposition in (\ref{eq:M_eigendecomp}), we get
\begin{align}
\!\!\!\!\!\!\|\widetilde{\mathbf{f}}_i\|^2\!=\!\mathbf{b}_i^H\!\mathbf{U}_i\big(\mathbf{\Lambda}_i\!+\!\mu_i\mathbf{I}\big)^{\!-\!2}\mathbf{U}_i^H\!\mathbf{b}_i\!=\!P_i. \label{eq:KKT_analysis_1}
\end{align}
By the definition of $\mathbf{p}_i$ in (\ref{eq:p_i_definition}), we rewrite (\ref{eq:KKT_analysis_1}) as
\begin{align}
\!\!\!\!\!\!\|\widetilde{\mathbf{f}}_i\|^2=g_i(\mu_i)=\sum_{k=1}^{r_i}\frac{|p_{i,k}|^2}{(\lambda_{i,k}+\mu_i)^2}\!+\!\!\!\!\sum_{k=r_i+1}^{J_iN_i}\!\!\!\frac{\ |p_{i,k}|^2}{\mu_i^2}=P_i.\label{eq:KKT_analysis_2}
\end{align}
Note that here $g_i(\mu_i)$ is a positive, continuous and strictly decreasing function in $\mu_i$. 

To identify the positivity of $\mu_i$, the following different cases are considered:
\newline
\underline{{$\mathsf{CASE}$ (\uppercase\expandafter{\romannumeral1})}}--- $\mu_i^{\star}>0$
\indent This case further involves two subcases:
\newline\indent \underline{case {\expandafter{\romannumeral1}})}--- $\exists k\in\{r_i+1,\cdots,J_iN_i\}$ s.t. $|p_{i,k}|\neq0$:

In this case, it is easily seen that $g_i(\mu_i)\rightarrow+\infty$ when $\mu_i\rightarrow0^+$, so $g_i(\mu_i)$ has the range of $\big(0,\infty\big)$ for positive $\mu_i$. So in case {\expandafter{\romannumeral1}}) there always exists a unique positive $\mu_i$ satisfying (\ref{eq:KKT_analysis_2}). Suppose that the unique solution of (\ref{eq:KKT_analysis_2}) is $\mu_i^{\star}$. Plugging $\mu_i^{\star}$ back into the KKT condition (\ref{eq:P2_i_KKT_a}), we obtain the optimal solution $\widetilde{\mathbf{f}}_i^{\star}$ as
\begin{align}
\widetilde{\mathbf{f}}_i^{\star}=\big(\mathbf{M}_i+\mu_i^{\star}\mathbf{I}\big)^{-1}\mathbf{b}_i.\label{eq:f_i_opt}
\end{align}
Plugging (\ref{eq:new_def}) into the above, (\ref{eq:opt_f_i_1}) is obtained. It is easily verified that the  $\mu_i^{\star}$ and $\mathbf{f}_i^{\star}$ in (\ref{eq:opt_f_i_1}) satisfy all the KKT conditions in (\ref{eq:P2_i_KKT}) and therefore is the optimal solution to ($\mathsf{P}2'_i$).
\newline\indent \underline{case {\expandafter{\romannumeral2}})} $\sum_{k=r_i+1}^{J_iN_i}|p_{i,k}|^2\!=\!0$ and $\sum_{k\!=\!1}^{r_i}\!\!\!\frac{\ |p_{i,k}|^2}{\lambda_{i,k}^2}\!>\!P_i$:
\newline\indent In this case, the second part in the summation of $g_i(\mu_i)$ in (\ref{eq:KKT_analysis_2}) vanishes and $g_i(\mu_i)$ has the bounded range $\Big(0,\sum_{k=1}^{r_i}\!\!\!\frac{\ |p_{i,k}|^2}{\lambda_{i,k}^2}\Big]$, with its maximum value achieved at $\mu_i=0$. When $\sum_{k=1}^{r_i}{|p_{i,k}|^2}{\lambda_{i,k}^2}>P_i$, a positive $\mu_i^{\star}$ satisfying (\ref{eq:KKT_analysis_2}) still exists and is unique. Consequently, the optimal solution $\mathbf{f}_i^{\star}$ can be determined by (\ref{eq:f_i_opt}) as in the subcase \expandafter{\romannumeral1}). 
\newline
\underline{$\mathsf{CASE}$ (\uppercase\expandafter{\romannumeral2})}--- $\sum_{k=r_i+1}^{J_iN_i}|p_{i,k}|^2\!=\!0$ and $\sum_{k\!=\!1}^{r_i}\!\!\!\frac{\ |p_{i,k}|^2}{\lambda_{i,k}^2}\!\leq\!P_i$

In this case, a positive $\mu_i$ satisfying KKT conditions does not exist any more, and $\mu_i^{\star}=0$. As such, the optimal solution $\mathbf{f}_i^{\star}$ should satisfy (\ref{eq:P2_i_KKT_a}): 
\begin{align}
\mathbf{M}_i\widetilde{\mathbf{f}}_i=\mathbf{b}_i.\label{eq:linear_sys}
\end{align}

We now claim that the above equation (\ref{eq:linear_sys}) has a feasible solution. Indeed, this equation is solvable if and only if the right hand side $\mathbf{b}_i$ belongs to the column space $\mathcal{R}\big(\mathbf{M}_i\big)$. Recall that $\mathbf{M}_i$ is Hermitian and has rank $r_i$; so $\mathcal{R}\big(\mathbf{M}_i\big)=\mathsf{span}\big(\mathbf{u}_{i,1},\cdots,\mathbf{u}_{i,r_i}\big)$ and the null space of $\mathbf{M}_i$ satisfies $\mathcal{N}\big(\mathbf{M}_i\big)=\mathcal{R}^{\perp}\big(\mathbf{M}_i\big)=\mathsf{span}\big(\mathbf{u}_{i,r_i\!+\!1},\cdots,\mathbf{u}_{i,J_iN_i}\big)$. In fact, $\mathbb{C}^{J_iN_i}=\mathcal{R}\big(\mathbf{M}_i\big)\oplus\mathcal{N}\big(\mathbf{M}_i\big)$. Invoking the assumption of CASE (\uppercase\expandafter{\romannumeral2}) that $|p_{i,k}|=0$, $\forall k\in\{r_i+1,\cdots,J_iN_i\}$ and the definition of $\mathbf{p}_i$, we obtain $p_{i,k}=\mathbf{u}_{i,k}^H\mathbf{b}_i$, $\forall k\in\{r_i+1,\cdots,J_iN_i\}$. Actually this implies $\mathbf{b}_i\in\mathcal{N}^{\perp}\big(\mathbf{M}_i\big)=\mathcal{R}\big(\mathbf{M}_i\big)$ and thus the consistency (i.e. the feasibility) of (\ref{eq:linear_sys}) is guaranteed.

Next we proceed to analytically identify one special feasible solution of (\ref{eq:linear_sys}). Eigenvalue decomposing $\mathbf{M}_i$, (\ref{eq:linear_sys}) can be equivalently written as
\begin{align}
\mathbf{\Lambda}_i\mathbf{U}_i^H\widetilde{\mathbf{f}}_i=\mathbf{p}_i.\label{eq:linear_sys_2}
\end{align}
Let $\bar{\mathbf{\Lambda}}_i$ represent the top-left $r_i\times r_i$ sub-matrix of $\mathbf{\Lambda}_i$, i.e. $\mathbf{\Lambda}_i=\mathsf{diag}\big\{\bar{\mathbf{\Lambda}}_i,\mathbf{O}_{J_iN_i-r_i}\big\}$. Let $\bar{\mathbf{U}}_i$ and $\tilde{\mathbf{U}}_i$ represent the left-most $r_i$ columns and the remaining columns of $\mathbf{U}_i$ respectively, i.e. $\mathbf{U}_i=\big[\bar{\mathbf{U}}_i,\tilde{\mathbf{U}}_i\big]$. We can then simplify (\ref{eq:linear_sys_2}) to
\begin{align}
\bar{\mathbf{\Lambda}}_i\bar{\mathbf{U}}_i^H\widetilde{\mathbf{f}}_i=\mathbf{p}_i.\label{eq:linear_sys_3}
\end{align}

Since the columns of $\mathbf{U}_i$ form a set of orthonormal basis for $\mathbb{C}^{J_iN_i}$, $\widetilde{\mathbf{f}}_i$ can be expressed via columns of $\mathbf{U}_i$ as $\widetilde{\mathbf{f}}_i=\sum_{k=1}^{J_iN_i}\alpha_{i,k}\mathbf{u}_{i,k}$. Noticing the key fact that $\bar{\mathbf{U}}_i^H\mathbf{u}_{i,k}=\mathbf{0}$, $\forall k\in\{r_i+1,\cdots,J_iN_i\}$, we know that the values of $\{\alpha_{i,r_i+1},\cdots,\alpha_{i,J_iN_i}\}$ have no impact on (\ref{eq:linear_sys_3}) and can therefore be safely set to zeros to save energy. As for $\alpha_{i,k}$, $\forall k\in\{1,\cdots,r_i\}$, we  substitute $\widetilde{\mathbf{f}}_i=\sum_{k=1}^{r_i}\alpha_{i,k}\mathbf{u}_{i,k}$ into (\ref{eq:linear_sys_3}) and obtain 
\begin{align}
\alpha_{i,k}=\lambda_{i,k}^{-1}p_{i,k},\ \ \ \ \forall k \in\{1, \cdots, r_i\}.\label{eq:coefficient_alpha}
\end{align}
Summarizing the above analysis, the optimal solution $\widetilde{\mathbf{f}}_i^{\star}$ to ($\mathsf{P}3_i$) is given by
\begin{align}
\widetilde{\mathbf{f}}_i^{\star}=\mathbf{U}_i\mathbf{\Lambda}_i^{\dagger}\mathbf{U}_i^H\mathbf{b}_i,\label{eq:f_i_opt_mu_i0}
\end{align}
with $\mathbf{\Lambda}_i^{\dagger}$ being the Moore-Penrose pseudoinverse of $\mathbf{\Lambda}_i$ given as $\mathsf{diag}\big\{\bar{\mathbf{\Lambda}}_i^{-1},\mathbf{O}_{J_iN_i\!-\!r_i}\big\}$. Matrix theory suggests that an arbitrary matrix $\mathbf{X}$ with its singular value decomposition (SVD) given by $\mathbf{X}=\mathbf{U}_{\mathbf{X}}\mathbf{\Lambda}_{\mathbf{X}}\mathbf{V}^H_{\mathbf{X}}$ has its unique Moore-Penrose pseudoinverse $\mathbf{X}^{\dagger}=\mathbf{V}_{\mathbf{X}}\mathbf{\Lambda}_{\mathbf{X}}^{\dagger}\mathbf{U}^H_{\mathbf{X}}$, where $\mathbf{U}_{\mathbf{X}}$ and $\mathbf{V}_{\mathbf{X}}$ are left and right singular square matrices, respectively, and $\mathbf{\Lambda}_{\mathbf{X}}$ is a diagonal matrix with appropriate dimensions. Hence, (\ref{eq:f_i_opt_mu_i0}) can be equivalently written as
\begin{align}
\widetilde{\mathbf{f}}_i^{\star}=\mathbf{M}_i^{\dagger}\mathbf{b}_i.\label{eq:f_i_opt_MPinverse}
\end{align}

Obviously $\mu_i^{\star}=0$, and $\mu_i^{\star}$ and $\widetilde{\mathbf{f}}_i^{\star}$ satisfy the KKT conditions (\ref{eq:P2_i_KKT_a}), (\ref{eq:P2_i_KKT_c}) and (\ref{eq:P2_i_KKT_d}). What remains to be shown is that $\widetilde{\mathbf{f}}_i^{\star}$  satisfies the power constraint. We verify this using (\ref{eq:coefficient_alpha}) and get 
\begin{align}
\|\widetilde{\mathbf{f}}_i^{\star}\|^2=\sum_{k=1}^{r_i}|\alpha_{i,k}|^2=\sum_{k=1}^{r_i}\frac{\ |p_{i,k}|^2}{\lambda_{i,k}^2}\leq P_i,
\label{eq:solution_case3}
\end{align}
where the inequality in the above follows the assumption of CASE (\uppercase\expandafter{\romannumeral2}). Plugging (\ref{eq:new_def}) into (\ref{eq:f_i_opt_MPinverse}), (\ref{eq:opt_f_i_2}) is obtained.
The proof is complete.
\end{proof}

\subsection{Proof of Theorem \ref{thm:converge_essential_cyclic}}
\label{subsec:thm_converge_essential_cyclic}

\begin{proof}
This proof is inspired by Proposition 2.7.1 in \cite{bib:nonlinear_programming_Bertsekas}. To simplify the following exposition, we define $\mathbf{x}\triangleq[\mathbf{x}_1^T,\cdots,\mathbf{x}_{L+1}^T]=[\mathbf{f}_1^T,\cdots,\mathbf{f}_L^T,\mathbf{g}^T]$ and $\mathbf{x}\in\mathcal{X}\triangleq\mathcal{X}_1\times\cdots\times\mathcal{X}_{L+1}$ with $\mathcal{X}_i=\big\{\mathbf{f}_i\in\mathbb{C}^{J_iN_i}\big|\mathbf{f}_i^H\mathbf{E}_i\mathbf{f}_i\!\leq\!P_i\big\}$, for $i=1,\cdots,L$ and $\mathcal{X}_{L+1}=\mathbb{C}^{KM}$. 
For any specific essentially cyclic update BCD algorithm, we assume that it starts from an initial feasible solution 
$\mathbf{x}^{(0)}\triangleq[\mathbf{x}_1^T{}^{(0)},\cdots,\mathbf{x}_{L+1}^T{}^{(0)}]$ and the iteration index $(k)$ increases by one after any block's update.
Denote $\mathbf{x}_{i}^{(k)}$ as the $i$-th block of $\mathbf{x}^{(k)}$ and $\mathbf{x}_{\bar{i}}=[\mathbf{x}_1,\cdots,\mathbf{x}_{i-1},\mathbf{x}_{i+1},\cdots,\mathbf{x}_{L+1}]$, $i\in\{1,\cdots,L+1\}$, $i\in\{1,\cdots,L+1\}$. Assume that $T$ is a period of the essentially cyclic update rule and $\{t_1,\cdots,t_T\}$, with $t_j\in\{1,\cdots,L+1\}$ $\forall j\in\{1,\cdots,T\}$, as the indices of the updated blocks in a period in order. If $\mathbf{x}_{t_j}$ is updated in the $(k)$-th iteration, then $\mathbf{x}_{t_{j\oplus1}}$ is updated in the $(k\!+\!1)$-th iteration. Define $j\!\oplus\!1\triangleq\!j(\!\!\!\!\mod T)\!+\!1$, $\forall j\in\{1,\cdots,T\}$ and $j\!\oplus\!m$ as $j\!\oplus\!1$ by $m$ times. 

By repeatedly invoking Bolzano-Weierstrass theorem to $\mathbf{f}_i$ to $\mathbf{f}_L$ and noticing that $\mathbf{g}$ is updated in closed form by equation (\ref{eq:G_MMSE}), the existence of limit points of $\{\mathbf{x}^{(k)}\}_{k=0}^{\infty}$ can be proved.   

Then we prove that $\mathsf{MSE}^{(k)}$ is decreasing. If $\mathbf{x}_{L+1}$(or $\mathbf{g}$) is updated in the $(k\!+\!1)$-th iteration, then ($\mathsf{P}1$) is solved and thus MSE is decreasing. Assume that in the $(k\!+\!1)$-th iteration, the ($t_{j\oplus1}$)-th block is updated, $t_{j\oplus1}\in\{1,\cdots,L\}$. Then 
\begin{align}
\mathbf{x}_{t_{j\oplus1}}^{(k+1)}\!\!=\!\!\underset{\mathbf{x}_{t_{j\oplus1}}\in\mathcal{X}_{t_{j\oplus1}}}{\arg\min.} \mathsf{MSE}\big(\mathbf{x}_{t_{j\oplus1}}\big|\mathbf{x}_{\overline{t_{j\oplus1}}}^{(k)}\big)\!\!+\!\!\kappa\big\|\mathbf{x}_{t_{j\oplus1}}\!-\!\mathbf{x}_{t_{j\oplus1}}^{(k)}\big\|_2^2. \nonumber
\end{align}
Since $\mathbf{x}_{t_{j}}^{(k)}$ is feasible, it should give no smaller objective than  $\mathbf{x}_{t_{j}}^{(k+1)}$ for the above problem. This implies
\begin{align}
\!\!\mathsf{MSE}\big(\mathbf{x}^{(k+1)}\big)\!\leq\!\mathsf{MSE}\big(\mathbf{x}^{(k)}\big)\!-\!\kappa\big\|\mathbf{x}^{(k)}\!\!-\!\!\mathbf{x}^{(k+1)}\big\|_2^2\!\leq\!\mathsf{MSE}\big(\mathbf{x}^{(k)}\big).\nonumber
\end{align}
Thus $\mathsf{MSE}^{(k)}$ is decreasing. At the same time notice that $\mathsf{MSE}$ should be nonnegative, thus $\mathsf{MSE}^{(k)}$ converges. 

Next we prove that any limit point is stationary. Assume that a subsequence of solution $\mathbf{x}^{(k_j)}$ converges to a limit point $\bar{\mathbf{x}}\triangleq[\bar{\mathbf{x}}_1^T,\cdots,\bar{\mathbf{x}}_{L+1}^T]$. Since there are 
finite blocks, we assume the block $i\in\{1,\cdots,L+1\}$ is updated infinitely many times and assume that $i=t_l$ for some $l\in\{1,\cdots,T\}$. It should be noted that such $l$ may be non-unique and arbitrary one can be chosen to do the job.

We assert that $\mathbf{x}^{(k_j\!+\!1)}\rightarrow\bar{\mathbf{x}}$, i.e. $\mathbf{x}^{(k_j+1)}_{t_{l\oplus1}}\rightarrow\bar{\mathbf{x}}_{t_{l\oplus1}}$.  This claim can be proved in two cases---i) $t_{l\oplus1}\!=\!L\!+\!1$ and ii) $t_{l\oplus1}\in\{1,\cdots,L\}$ .

i) $t_{l\oplus1}\!=\!L\!+\!1$. Notice that $\mathbf{x}_{L+1}=\mathbf{g}$ is updated in a closed form (\ref{eq:G_MMSE}), which is a continuous function of $[\mathbf{x}_1^T,\cdots,\mathbf{x}_L^T]$. Since 
$\mathbf{x}_{\overline{L\!+\!1}}^{(k_j)}$ converges, by taking $j\rightarrow\infty$, $\mathbf{x}_{t_{l\oplus1}}^{(k_j+1)}$ should converge to some limit, i.e. $\mathbf{x}_{t_{l\oplus1}}^{(k_j+1)}\rightarrow\widetilde{\mathbf{x}}_{L+1}$. Notice that $\mathsf{MSE}^{(k)}$ converges, so $\mathsf{MSE}\big(\bar{\mathbf{x}}_{\overline{L+1}},\bar{\mathbf{x}}_{L+1}\big)=\mathsf{MSE}\big(\bar{\mathbf{x}}_{\overline{L+1}},\widetilde{\mathbf{x}}_{L+1}\big)$. This means both $\bar{\mathbf{x}}_{L+1}$ and $\widetilde{\mathbf{x}}_{L+1}$ are solutions to the problem ($\mathsf{P}1$) with sensors' beamformers $[\bar{\mathbf{x}}_1^T,\cdots,\bar{\mathbf{x}}_L^T]$ given. Since ($\mathsf{P}1$) is strictly convex and thus has unique solution, we conclude $\widetilde{\mathbf{x}}_{L+1}=\bar{\mathbf{x}}_{L+1}$. So 
$\mathbf{x}^{(k_j\!+\!1)}_{t_{l\oplus1}}\rightarrow\bar{\mathbf{x}}_{t_{l\oplus1}}$ holds for the case $t_{l\oplus1}=L+1$.

ii) $t_{l\oplus1}\in\{1,\cdots,L\}$. By contradiction, we assume that $\mathbf{x}^{(k_j+1)}_{t_{l\oplus1}}$ does not converge to $\bar{\mathbf{x}}_{t_{l\oplus1}}$.
By denoting $\gamma^{(k_j)}\triangleq\|\mathbf{x}^{(k_j+1)}_{t_{l\oplus1}}\!-\!\bar{\mathbf{x}}_{t_{l\oplus1}}\|_2$ and possibly restricting to a subsequence, we assume that there exists a $\bar{\gamma}>0$ such that $\gamma^{(k_j)}\geq\bar{\gamma}$ for all $j$. Let $\mathbf{s}^{(k_j)}_{l}\!=\!(\mathbf{x}^{(k_j+1)}_{t_{l\oplus1}}\!-\!\mathbf{x}_{t_{l\oplus1}}^{(k_j)})/\gamma^{(k_j)}$. Since $\mathbf{s}^{(k_j)}$ is bounded, by Bolzano-Weierstrass theorem and restricting to a subsequence, we assume that $\mathbf{s}^{(k_j)}\rightarrow\bar{\mathbf{s}}$. Then we obtain
\begin{align}
\!\!\!\!&\mathsf{MSE}\big(\mathbf{x}^{(k_j+1)}\big)\!=\!\mathsf{MSE}\big(\mathbf{x}^{(k_j+1)}_{t_{l\oplus1}}\big|\mathbf{x}^{(k_j)}_{\overline{t_{l\oplus1}}}\big) \\
\!\!\!\!\leq& \mathsf{MSE}\big(\mathbf{x}^{(k_j+1)}_{t_{l\oplus1}}\big|\mathbf{x}^{(k_j)}_{\overline{t_{l\oplus1}}}\big)\!+\!\kappa\big\|\mathbf{x}_{t_{l\oplus1}}^{(k_j+1)}\!-\!\mathbf{x}_{t_{l\oplus1}}^{(k_j)}\big\|_2^2\\
\!\!\!\!=&\mathsf{MSE}\big(\mathbf{x}^{(k_j)}_{t_{l\oplus1}}\!+\!\gamma^{(k_j)}\mathbf{s}^{(k_j)}\big|\mathbf{x}^{(k_j)}_{\overline{t_{l\oplus1}}}\big)\!+\!\kappa\big\|\gamma^{(k_j)}\mathbf{s}^{(k_j)}\big\|_2^2\\
\!\!\!\!\leq&\mathsf{MSE}\big(\mathbf{x}^{(k_j)}_{t_{l\oplus1}}\!+\!\epsilon\bar{\gamma}\mathbf{s}^{(k_j)}\big|\mathbf{x}^{(k_j)}_{\overline{t_{l\oplus1}}}\big)\!+\!\kappa\big\|\epsilon\bar{\gamma}\mathbf{s}^{(k_j)}\big\|_2^2,\forall\epsilon\!\in\![0,1] \label{eq:proof_ineq_1}\\
\!\!\!\!\leq&\mathsf{MSE}\big(\mathbf{x}^{(k_j)}_{t_{l\oplus1}}\big|\mathbf{x}^{(k_j)}_{\overline{t_{l\oplus1}}}\big)=\mathsf{MSE}\big(\mathbf{x}^{(k_j)}\big),\label{eq:proof_ineq_2}
\end{align}
where the last two inequalities follow the fact that $\mathsf{MSE}\big(\mathbf{x}_{t_{l\oplus1}}\big|\mathbf{x}^{(k_j)}_{\overline{t_{l\oplus1}}}\big)\!+\!\kappa\|\mathbf{x}_{t_{l\oplus1}}\!\!-\!\!\mathbf{x}_{t_{l\oplus1}}^{(k_j)}\|_2^2$ is strictly convex and attains the minimum at point $\mathbf{x}^{(k_j+1)}_{t_{l\oplus1}}$. Noting $\mathsf{MSE}^{(k_j)}$ converges and letting $j\rightarrow\infty$, we obtain
\begin{align}
\!\!\!\!\!\!\mathsf{MSE}\big(\bar{\mathbf{x}}\big)&\leq\mathsf{MSE}\big(\bar{\mathbf{x}}_{t_{l\oplus1}}\!\!+\!\epsilon\bar{\gamma}\bar{\mathbf{s}}\big|\bar{\mathbf{x}}_{\overline{t_{l\oplus1}}}\big)\!+\!\kappa\epsilon^2\bar{\gamma}^2\nonumber\\
&\leq\mathsf{MSE}\big(\bar{\mathbf{x}}\big), \quad \forall\epsilon\in[0,1],
\end{align}
which immediately implies 
\begin{align}
\!\!\!\!\!\!\!\!\mathsf{MSE}\big(\bar{\mathbf{x}}_{t_{l\oplus1}}\!\!+\!\epsilon\bar{\gamma}\bar{\mathbf{s}}\big|\bar{\mathbf{x}}_{\overline{t_{l\oplus1}}}\big)\!+\!\kappa\epsilon^2\bar{\gamma}^2\!=\!\mathsf{MSE}\big(\bar{\mathbf{x}}\big),\forall\epsilon\in[0,1].\label{eq:proof_strict_convex_quadratic}
\end{align}
However the above is impossible. Notice that $\mathsf{MSE}\big(\bar{\mathbf{x}}_{t_{l\oplus1}}\!\!+\!\!\epsilon\bar{\gamma}\bar{\mathbf{s}}\big|\bar{\mathbf{x}}_{\overline{t_{l\oplus1}}}\big)$ is a quadratic function of $\epsilon$ with nonnegative quadratic coefficient and  $\bar{\gamma},\kappa>0$. Thus the left hand side(LHS) of equation (\ref{eq:proof_strict_convex_quadratic}) is a strictly convex quadratic function of $\epsilon$, which has at most two different $\epsilon$ giving the function value of $\mathsf{MSE}(\bar{\mathbf{x}})$. Contradiction has been reached. 

In the above we have proved that $\mathbf{x}^{(k_j+1)}\rightarrow\bar{\mathbf{x}}$. Next we show that $\mathbf{\nabla}_{\mathbf{x}_{t_{l\oplus1}}}\mathsf{MSE}\big(\bar{\mathbf{x}}\big)^T\big(\mathbf{x}_{t_{l\oplus1}}\!-\!\bar{\mathbf{x}}_{t_{l\oplus1}}\big)\geq0$, $\forall\mathbf{x}_{t_{l\oplus1}}\in\mathcal{X}_{t_{l\oplus1}}$, which is also proved in two cases:

When $t_{l\oplus1}\in\{1,\cdots,L\}$, we have
\begin{align}
\!\!\!\mathbf{x}^{(k_j\!+\!1)}_{t_{l\oplus\!1}}\!\!=\!\!\underset{\mathbf{x}_{t_{l\oplus1}}\!\in\!\mathcal{X}_{t_{l\oplus1}}}{\arg\min.}\mathsf{MSE}\Big(\mathbf{x}_{t_{l\!\oplus\!1}}\big|\mathbf{x}^{(k_j)}_{\overline{t_{l\oplus1}}}\Big)\!+\!\kappa\big\|\mathbf{x}_{t_{l\oplus1}}\!-\!\mathbf{x}_{t_{l\oplus1}}^{(k_j)}\!\big\|_2^2. \nonumber
\end{align}
By optimality condition, the above implies
\begin{align}
&\mathbf{\nabla}_{\mathbf{x}_{t_{l\oplus1}}}\mathsf{MSE}\big(\mathbf{x}^{(k_j+1)}_{t_{l\oplus1}}\big|\mathbf{x}^{(k_j)}_{\overline{t_{l\oplus1}}}\big)^T\big(\mathbf{x}_{t_{l\oplus1}}\!-\!\mathbf{x}^{(k_j+1)}_{t_{l\oplus1}}\big), \\
&\!+\!2\kappa\big(\mathbf{x}_{t_{l\oplus1}}^{(k_j+1)}\!-\!\mathbf{x}_{t_{l\oplus1}}^{(k_j)}\big)^T\big(\mathbf{x}_{t_{l\oplus1}}\!-\!\mathbf{x}^{(k_j+1)}_{t_{l\oplus1}}\big)\geq0, \forall\mathbf{x}_{t_{l\oplus1}}\in\mathcal{X}_{t_{l\oplus1}}. \nonumber
\end{align}
Let $j\rightarrow\infty$ in the above equation and note that $\mathsf{MSE}$ is continuously differentiable, we obtain
\begin{align}
\mathbf{\nabla}_{\mathbf{x}_{t_{l\oplus1}}}\mathsf{MSE}\big(\bar{\mathbf{x}}\big)^T\big(\mathbf{x}_{t_{l\oplus1}}\!-\!\bar{\mathbf{x}}_{t_{l\oplus1}}\big)\geq0, \forall\mathbf{x}_{t_{l\oplus1}}\in\mathcal{X}_{t_{l\oplus1}}. 
\end{align}

When $t_{l\oplus1}=L+1$, the above reasoning still works except that the proximal term is absent (i.e. $\kappa=0$). So we also obtain $\mathbf{\nabla}_{\mathbf{x}_{t_{l\oplus1}}}\mathsf{MSE}\big(\bar{\mathbf{x}}\big)^T\big(\mathbf{x}_{t_{l\oplus1}}\!-\!\bar{\mathbf{x}}_{t_{l\oplus1}}\big)\geq0$.
%

Now replace the subsequence $\{k_j\}$ with $\{k_j+1\}$, $t_{l\oplus1}$ with $t_{l\oplus2}$ and utilize the verbatim argument as above, we can prove  
\begin{align}
\mathbf{\nabla}_{\mathbf{x}_{t_{l\oplus2}}}\mathsf{MSE}\big(\bar{\mathbf{x}}\big)^T\big(\mathbf{x}_{t_{l\oplus2}}\!-\!\bar{\mathbf{x}}_{t_{l\oplus2}}\big)\geq0, \forall\mathbf{x}_{t_{l\oplus2}}\in\mathcal{X}_{t_{l\oplus2}}. 
\end{align}
Repeating this argument for $(T-1)$ times and recalling that for essentially cyclic update rule, $\{t_{l\oplus1},\cdots,t_{l\oplus T}\}\!=\!\{1,\cdots,L\}$, we have proved that 
\begin{align}
\!\!\!\mathbf{\nabla}_{\mathbf{x}_{i}}\!\mathsf{MSE}\big(\bar{\mathbf{x}}\big)^T\!\big(\mathbf{x}_{i}\!-\!\bar{\mathbf{x}}_{i}\big)\!\geq\!0, \forall\mathbf{x}_{i}\in\mathcal{X}_{i},\forall i\!\in\!\{\!1,\!\cdots\!,\!L\!\!+\!\!1\!\}.   
\end{align}
Summing up the above ($L\!+\!1$) inequalities, we obtain 
\begin{align}
\!\!\!\mathbf{\nabla}_{\mathbf{x}}\mathsf{MSE}\big(\bar{\mathbf{x}}\big)^T\!\big(\mathbf{x}\!-\!\bar{\mathbf{x}}\big)\!\geq\!0, \forall\mathbf{x}\in\mathcal{X}.   
\end{align}
So $\bar{\mathbf{x}}$ is actually a stationary point of ($\mathsf{P}0$).
\end{proof}

\end{document}